%% file: mq.tex
\documentclass[msom,nonblindrev]{informs3}
\renewcommand{\theARTICLETOP}{}
\renewcommand{\theRRHSecondLine}{}
\renewcommand{\theLRHSecondLine}{}

\usepackage{tikz}
\usepackage{pgfplots}
\usepgfplotslibrary{fillbetween}
\usetikzlibrary{positioning}
\usepackage{tkz-euclide}
\usetikzlibrary{arrows.meta}

\usepackage{color, xcolor} 
\colorlet{darkblue}{blue!40!black}
\definecolor{auburn}{rgb}{0.43, 0.21, 0.1}

\usepackage{subcaption} %
\usepackage{caption}
\usepackage{lmodern}
\usepackage{graphicx}
\usepackage{mathtools}

\newcommand{\OPT}[0]{\textrm{OPT}}
\newcommand{\E}{\mathbb{E}} %
\newcommand{\PP}{\mathbb{P}} %

\newcommand{\gettikzxy}[3]{%
  \tikz@scan@one@point\pgfutil@firstofone#1\relax
  \edef#2{\the\pgf@x}%
  \edef#3{\the\pgf@y}%
}

\newcommand{\vecAgent}[0]{A}
\newcommand{\agent}[0]{A}
\newcommand{\vecQueue}[0]{Q}
\newcommand{\queue}[0]{Q}

\newcommand{\typepolicy}[0]{\phi}
\newcommand{\typepolicyset}[0]{\Phi}

\newcommand{\vecQ}[0]{A}

\newcommand{\vece}[0]{e}

\newcommand{\RND}{\textrm{RND}}

\newcommand{\renege}{\theta} %
\newcommand{\agentarr}{\lambda} %
\newcommand{\jobarr}{\mu}

\newcommand{\joinprob}{\sigma}
\newcommand{\joinprobset}{\Sigma}

\newcommand{\joinprobFLQ}{\sigma^{\textrm{FR}}}
\newcommand{\joinprobFLQR}{\sigma^{\textrm{FRfb}}}

\newcommand{\policy}{\pi}

\newcommand{\queueset}{\mathcal{Q}}

\newcommand{\SMQ}[0]{\policy^{\textrm{RND}}}
\newcommand{\FLQ}[0]{\policy^{\textrm{FR}}}

\newcommand{\SMQQ}[0]{\queueset^{\textrm{RND}}}
\newcommand{\FLQQ}[0]{\queueset^{\textrm{FR}}}

\newcommand{\FLQR}[0]{\policy^{\textrm{FRfb}}}

\newcommand{\agentset}{\mathcal{I}}
\newcommand{\loc}{\mathcal{J}}
\newcommand{\numLoc}{J}
\newcommand{\numAgent}{I}

\renewcommand{\th}{^{\mathrm{th}}}

\newcommand{\orderedqueue}{\rho}
\newcommand{\RCR}{\textrm{FRfb}}
\newcommand{\ACR}{\textrm{FR}}

\newcommand{\supt}{^{(t)}}
\newcommand{\supZero}{^{(0)}}

\newcommand{\setZ}{\mathbb{Z}}

\OneAndAHalfSpacedXI %

\usepackage{endnotes}
\usepackage{algorithm}
\usepackage{algpseudocode}
\usepackage{bm}
\usepackage{dsfont}
\usepackage{graphicx}

\usepackage{array}
\usepackage{makecell}

\usepackage{amsmath, amssymb}
\usepackage{mathrsfs}
\usepackage{comment}
\usepackage{mhequ}
\usepackage[resetlabels]{multibib}
\usepackage{pgfplotstable}
\newcites{ECA}{References}
\usepackage{cleveref}
\renewcommand{\be}{\begin{equs}}
\renewcommand{\ee}{\end{equs}}
\pgfplotsset{compat=1.12} %
\usepackage{multirow}
\usepackage{booktabs,caption,subcaption}
\usepackage{tikz}
\usepackage{pgfplots}
\usetikzlibrary{positioning,shapes,arrows,shadows,patterns}
\usepackage[flushleft]{threeparttable}

\usepackage{natbib}
 \bibpunct[, ]{(}{)}{,}{a}{}{,}%
 \def\bibfont{\small}%

\TheoremsNumberedThrough     %
\ECRepeatTheorems

\EquationsNumberedThrough    %

\begin{document}

\RUNAUTHOR{Yan et al.} %

\RUNTITLE{Matching Queues, Flexibility and Incentives}

\TITLE{Matching Queues, Flexibility and Incentives}

\ARTICLEAUTHORS{%
\AUTHOR{Chiwei Yan}
\AFF{Department of Industrial Engineering and Operations Research, University of California, Berkeley \EMAIL{}}
\AUTHOR{Francisco Castro}
\AFF{Decisions, Operations and Technology Management, UCLA Anderson School of Management  \EMAIL{}} %
\AUTHOR{Peter Frazier}
\AFF{Operations Research and Information Engineering, Cornell University\EMAIL{}}
\AUTHOR{Hongyao Ma}
\AFF{Decision, Risk, and Operations, Columbia Business School  \EMAIL{}}
\AUTHOR{Hamid Nazerzadeh}
\AFF{Uber Technologies, Inc. \EMAIL{}}
} %

\ABSTRACT{%

\noindent\textbf{\emph{Problem definition}}: In many matching markets, some agents are fully flexible, while others only accept a subset of jobs. For example, ridesharing drivers can specify on the platform the destinations they are willing to accept. Conventional wisdom suggests reserving flexible agents, but this can backfire: anticipating higher matching chances, agents may misreport as specialized, reducing overall matches. We ask how platforms can design simple matching policies that remain effective when agents act strategically.
\textbf{\emph{Methodology/results}}: We model job allocation as a bipartite matching queueing system and analyze equilibrium throughput performance under different policies when agents choose which queue to join. We show that flexibility reservation is optimal under full information but can perform poorly with private information, sometimes substantially worse than random assignment. To address this, we propose a new policy---\emph{flexibility reservation with fallback}---that guarantees robust performance across settings, without requiring precise knowledge of system parameters or agent utility functions.
\textbf{\emph{Managerial implications}}: Our results underscore the importance of accounting for strategic reporting in the design of matching policies: the proposed fallback policy both preserves flexibility and exploits latent flexibility when explicitly flexible agents are exhausted. Its simplicity and parameter-free nature also make it practical to implement in platforms such as ridesharing and affordable housing allocation.

}%

\HISTORY{First version, 06/2020. Revision, 01/2024, 08/2025, 01/2026.}

\maketitle

\section{Introduction} \label{sec:introduction}

Matching markets play a critical role in business and society, coordinating the interaction of demand and supply. %
A central challenge that emerges in designing matching policies is to account for the private, heterogeneous preferences of strategic agents while maintaining operational efficiency. %
In a typical matching market, agents %
need to be %
matched with jobs to which they are compatible.
Some agents are highly flexible
and can accept most jobs, while others may be more specialized and can only be matched with a reduced subset.
However, if matching policies fail to balance opportunities appropriately, flexible agents may have incentives to misreport their type. Such behavior can lead to an effective loss in flexibility and ultimately reduce system efficiency.
The following examples illustrate how this dynamic can hurt matching performance and emphasize the importance of designing matching policies that align agents’ incentives with system objectives.

\smallskip
\textbf{Motivating Applications. } 
Ridesharing platforms such as Uber and Lyft,  
offer drivers %
the option to specify a destination area toward which they are willing to take trips. \Cref{fig:uber_1} and \Cref{fig:uber_2} illustrate how a driver specifies a desired destination area using Uber's ``driver destinations'' feature, opting out of trips heading the opposite directions \citep{destination_uber}. %
For drivers who need to head home or towards social or vocational obligations, 
this feature allows them to continue to earn on the platform instead of going offline. 
This provides substantial benefits for both sides: drivers earn more, and a larger driver pool leads to shorter pick-up times and better reliability for riders.

\begin{figure}[t!]
\centering
\begin{subfigure}[b]{0.31\textwidth}
    \centering
    \includegraphics[width=0.7\linewidth]{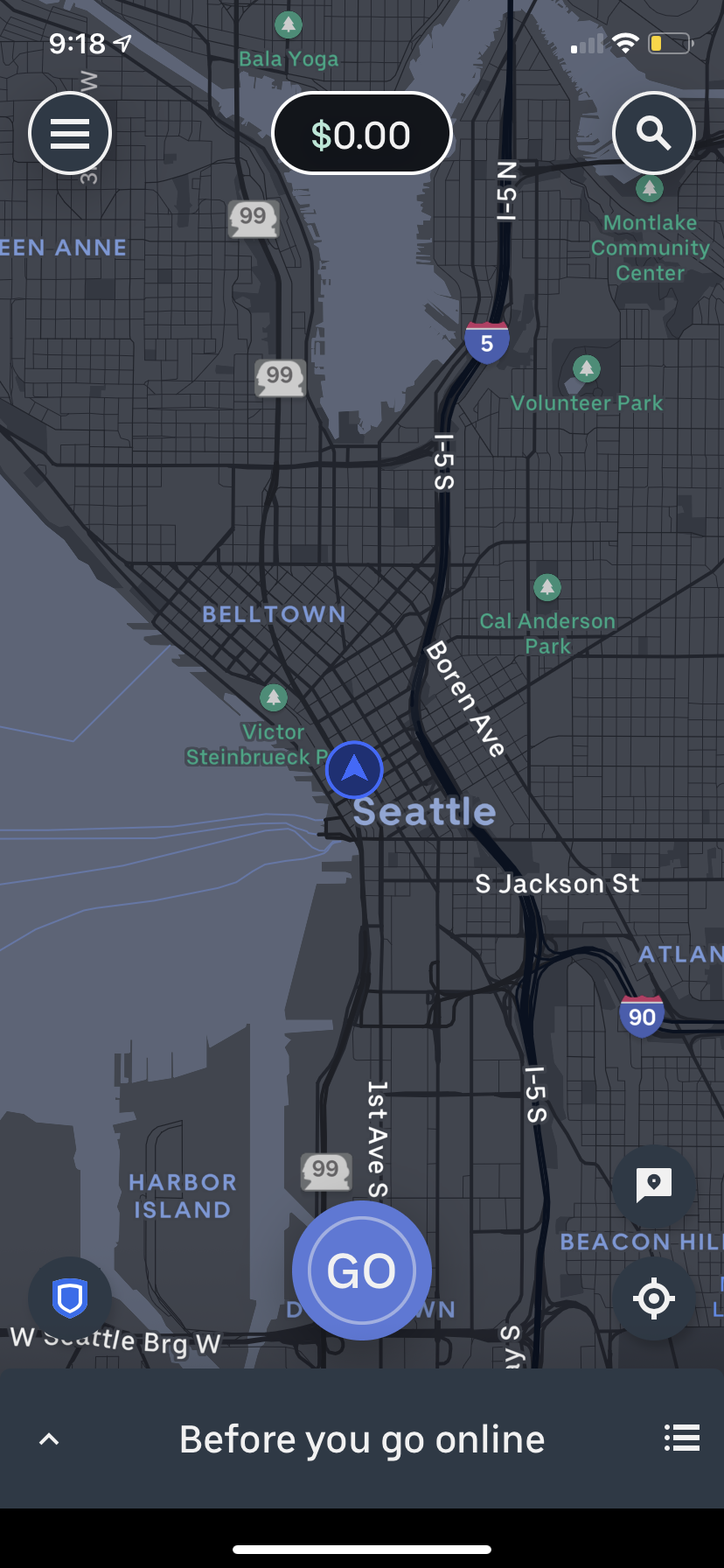}
    \caption{A driver turns online}
    \label{fig:uber_1}
\end{subfigure}
~
\begin{subfigure}[b]{0.31\textwidth}
  \centering
  \includegraphics[width=0.7\linewidth]{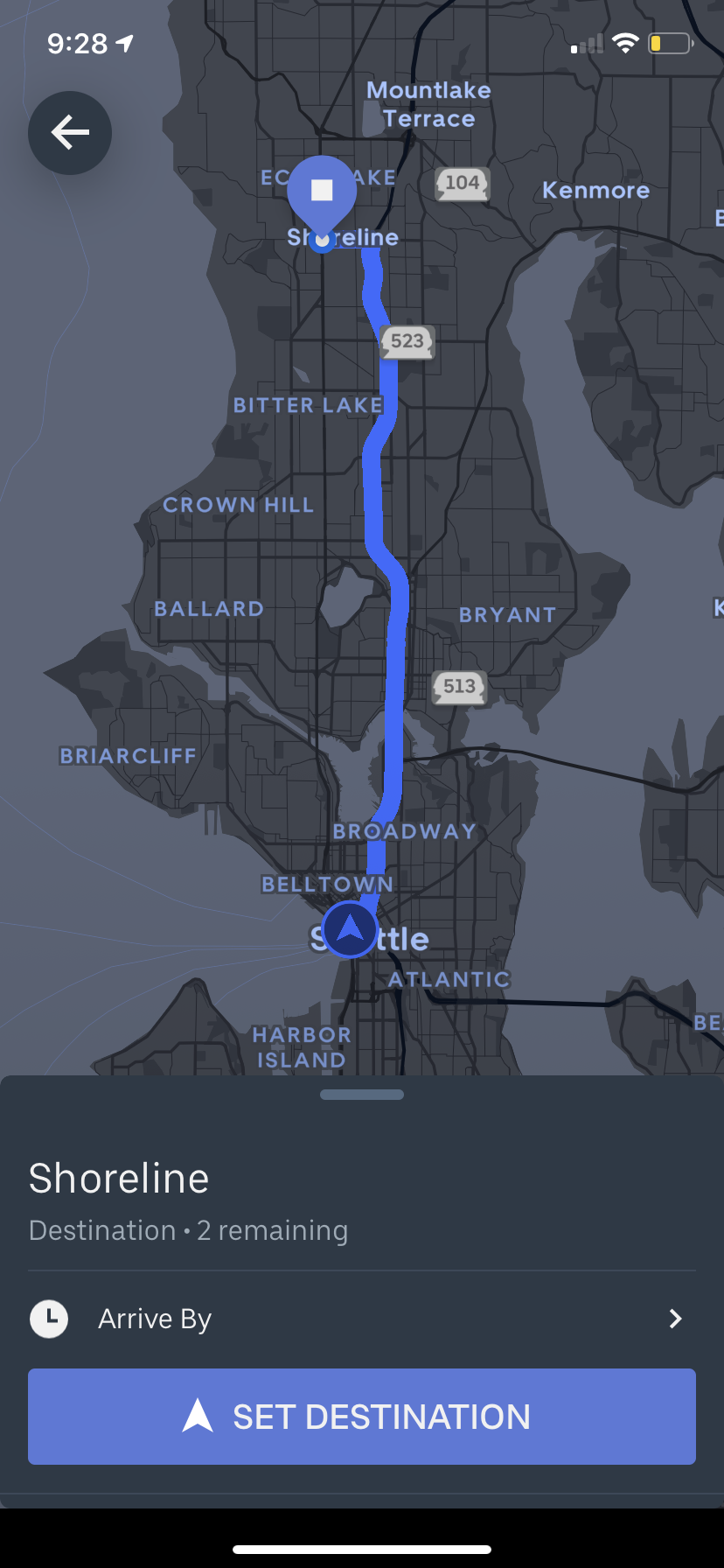}
  \caption{A driver sets her destination}
  \label{fig:uber_2}
\end{subfigure}
~
\begin{subfigure}[b]{0.31\textwidth}
  \centering
  \includegraphics[width=0.745\linewidth]{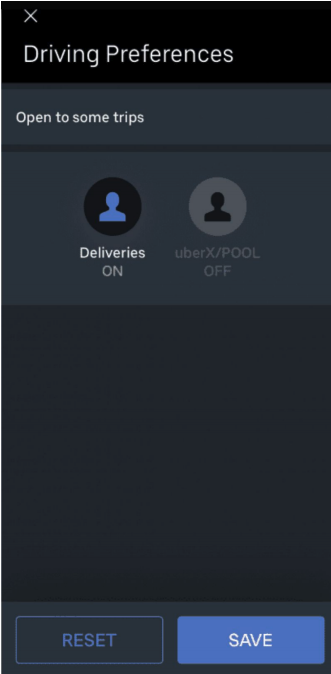}
  \caption{A driver sets job preference}
  \label{fig:uber_3}
\end{subfigure}
\caption{Uber's driver preferences feature. \normalfont{\Cref{fig:uber_1} is what a driver sees when she is online in downtown Seattle; \Cref{fig:uber_2} shows that the driver is using the destination product to accept trips toward Shoreline, a residential area in the north of Seattle; \Cref{fig:uber_3} shows the selection of a driver who chooses to take UberEats (food delivery) jobs but not Uber rides jobs.}}
\label{fig:driver_app}
\end{figure}

If the platform knew exactly drivers' flexibility over trip destinations, conventional wisdom would implement a matching policy that reserves flexible drivers as much as possible. 
In other words, a trip toward a specific destination would first be dispatched to a specialized driver who could only serve that destination. %
If there is no such driver, the trip would then be dispatched to flexible drivers who can serve more destinations. 
As a consequence of such %
reservation of flexible supply, if specialized drivers experience shorter waiting times to receive a request and have a better matching likelihood, flexible drivers may be incentivized to pretend to be specialized.
Such strategic behavior results in a ``loss of flexibility,'' and could %
degrade overall system performance relative to not reserving flexibility, casting doubt on the benefits of introducing products such as Uber's ``driver destinations" feature.

For example, in mid-2017, as part of a major driver campaign, Uber revised its destination feature, raising the daily limit from two to six uses per driver \citep{uber180}.
This increase allowed some drivers to use the feature more strategically \citep{Campbell2020Arrival}---for instance, to secure trips toward an airport%
---resulting in longer waiting and pick-up times for customers, and ultimately reducing the platform’s overall reliability \citep{Campbell2017}.
In response, Uber adjusted its matching system and reverted the daily limit from six back to two \citep{Etherington.2017}. To further mitigate this problem, the platform now advises drivers on its website that ``when a lot of drivers in a specific area set a destination, it limits the number of drivers that are able to accept trip requests from all riders. This causes longer wait times for both drivers and riders'' \citep{UberHelpWhyCantISetDestination}.\footnote{Lyft has similarly reduced its destination feature limit from six to two \citep{Cradeur2020LyftDestinationMode}.}

As a second application, consider gig-economy workers opting in for different types of jobs on the same platform. For example, Uber's driver partners can decide to receive UberEats (on-demand food delivery) jobs, or Uber rides jobs, or both (see \Cref{fig:uber_3}). %
Drivers often are capable of doing both types of jobs, though private information, such as whether a driver is accompanied by a friend in her car on a particular day, may prevent them from doing a certain job. To ensure a higher service level, it is preferable if the flexible drivers make themselves dispatchable for both types of jobs. However, flexible drivers may have an incentive to restrict themselves to only one type of job if doing so gives them higher priority for those jobs and reduces their waiting times. We refer interested readers to popular posts from online forums (see e.g., \citealt{UberXEat31:online, uberpeopledest,Separate16:online}) where experienced drivers discuss strategies of misreporting their flexibility to gain priority for certain jobs. %

The use of monetary incentives can potentially alleviate these strategic concerns, but in many practical settings, monetary incentives  %
are restricted by regulation, or business constraints. For example, in 2019 Uber experimented with a 30\% earning penalty on drivers using destination mode, and evenly distributed these reduced earnings to drivers not on destination mode \citep{uber30}. The driver feedback was negative (see driver discussions in \citealt{uber30driver}). 
In addition, paying drivers differently based on their levels of flexibility does not align with the changing regulatory environment 
(see \citealt{CaliforniaAB5} and \citealt{CaliforniaProp22Official}). Ridesharing platforms are moving towards giving driver partners more flexibility and transparency, including sharing trip information upfront, as well as providing drivers the options to accept or decline any trips \emph{without} any penalties.

A similar flexible-specialized structure arises in the public housing sector. However, the allocation mechanism in this setting creates inefficiencies that differ from those in the gig-economy platforms previously mentioned. %
Consider affordable housing allocation %
platforms such as NYC Housing Connect,\footnote{\url{https://housingconnect.nyc.gov/PublicWeb/about-us}, accessed on 08/11/2025.} a portal that prospective homeowners or renters can use to find housing. %
When a prospective homeowner (or renter) creates an account, they must declare whether they want to be considered for resales---existing units that become available when someone moves out---or re-rentals. After this, they can select listings of new units that match their preferences and apply. Units are then allocated via a lottery. Whenever a resale unit becomes available, all agents who expressed compatibility with this type of unit, even if they also expressed interest in new units, participate in the lottery. New units are accessible only to those who expressed interest in this type of unit. Because agents can decline offers, they might have an incentive to simply report preferences for both types of units, making the allocation mechanism like random assignment. This has contributed to a documented vacancy problem\footnote{\url{https://thenyhc.org/2025/02/10/nyhc-analysis-housing-connect-re-rental-vacancy-problems/}, accessed on 08/11/2025.} and suggests that alternative allocation rules, particularly those that make better use of flexible agents, could possibly yield more successful allocations.

\smallskip

\textbf{Research Questions.} 
The examples above highlight important problems that platforms face when designing matching policies: managing flexibility efficiently while accounting for agents’ potential strategic misreporting under private information. These challenges are further compounded by the fact that matching markets evolve dynamically over time, which calls for policies that perform robustly across diverse market conditions. In this context, we ask: (i) What principles should guide the design of matching policies under private information? (ii) How can we mitigate the loss of flexibility that arises from strategic behavior? (iii) Can we develop a simple and implementable, robust policy that performs well relative to natural benchmarks across settings?

\smallskip

\textbf{Model. } %
We study a game-theoretic queueing model, where jobs (e.g., trips in the context of ridesharing) and agents (e.g., drivers) arrive over time according to general arrival processes. 
Each job is associated with one of a finite number of types (e.g., %
trip destinations, or delivery jobs versus ride jobs), and an agent can be of the flexible type---compatible with all job types---or of the specialized type---compatible with a single job type.
The platform organizes a set of queues. Agent type is private information and agents report (or misreport) their types by joining the queue that maximizes their utility in equilibrium. 
Agents have exponentially distributed patience levels, and abandon the system if they have not been matched %
when their patience is exhausted. 
Jobs are assigned to agents in queues %
and agents are able to see the types of the jobs dispatched to them, and are %
\emph{free} to decline jobs without penalty.
However, jobs have \emph{limited patience} for repeated rejections and may be lost after excessive rejections. We study the impact of matching policies on the system's throughput, i.e., the number of matches per unit of time. %

\smallskip

\textbf{Contributions.} We now summarize our key contributions. %
\begin{itemize}
    \item We first confirm our intuition by showing that %
    when agents' types are known, the {\em full-information first-best} matching policy for throughput is achieved by reserving flexibility whenever possible (\Cref{prop:fb}) as flexible agents are more ``valuable'' for future matches. We call the optimal policy satisfying this property the flexibility reservation (FR) policy. %

    \item When agents are strategic and their types are private information, we show that such flexibility reservation may lead to longer waiting times and reduced matching likelihood for flexible agents. This incentivizes them to under-report the set of jobs they are compatible with, and this loss of flexibility may lower the system throughput. %
    In \Cref{thm:flip1-alt}, we show that the equilibrium throughput under the flexibility reservation policy can become \emph{substantially} bad compared to a compatibility-agnostic random matching policy (RND)---the effective baseline policy in the ridesharing and affordable housing examples.

    \item  %
    To balance matching efficiency with agents' strategic considerations, we propose a new policy dubbed {\em flexibility reservation with fallback} (FRfb). Intuitively, the FRfb policy retains a similar structure to the FR policy (i.e., reserving more flexible agents when possible), but offers additional {\em seemingly incompatible edges} along which jobs can be dispatched. In particular, when there is \emph{no} available compatible agents to match, a job will be sent to a pool of seemingly incompatible agents since some of them might be under-reporting and are, in fact, compatible with the job. 
    In contrast to the potential fragility of the FR policy, we show that, %
    the proposed FRfb policy enjoys \emph{robust} performance improvement. In particular, under any market conditions, and regardless of the strategy profile taken by the flexible agents, the FRfb policy \emph{always} achieves higher throughput than the random policy (\Cref{thm:robust_improvement_FLQR}). %
    We further demonstrate its performance via extensive simulations over general compatibility graphs. This robust performance guarantee, along with its parameter-free nature, makes our FRfb policy easy to implement in practice. 
    In particular, we illustrate how this policy is implemented in the driver destination product of a major ridesharing platform.

\end{itemize}

\smallskip

\textbf{Organization of the paper.} In \Cref{sec:lit}, we discuss related work. In 
\Cref{sec:preliminaries}, we introduce the main model elements. In \Cref{sec:fr}, we study the full-information optimal matching policy and benchmark its performance in the strategic setting.
In \Cref{sec:flq_with_recourse}, we propose and analyze the flexibility reservation policy with fallback. We conduct extensive simulation experiments in \Cref{sec:sim_experiments} to demonstrate the performance of various policies. We conclude and discuss a real-world implementation in \Cref{sec:conclusion}. %
We also extend our results to general compatibility graphs in Appendix \ref{sec:general_compatibility_graph} where each agent type can fulfill a certain subset of job types. All proofs and auxiliary technical results are presented in the online supplement.

\section{Literature Review} \label{sec:lit}

In this section, we review related work, focusing on three main streams: performance evaluation of skill-based queueing systems, queueing games, and non-monetary mechanism design.

\textbf{Skill-based queueing systems.} Our queueing-theoretic modeling framework generally falls under the skill-based server models. In these models, servers are flexible in the types of customers they can serve, and customers are flexible in the servers at which they can be processed. These models are motivated by settings such as call centers, where service representatives may speak different languages. Under the first-come first-served (FCFS) service discipline, various variants of the model exhibit product-form steady-state distributions. Notable works include \cite{adan2009exact}, \cite{visschers2012product} and \cite{adan2014skill}. We also refer readers to a recent excellent overview by \cite{gardner2020product} which synthesizes various related technical results in the field. Our modeling framework differs from this stream of literature in that we focus on random queueing discipline instead of FCFS, motivated by quality-driven dispatch used in ridesharing platforms as well as lotteries used in allocating affordable housing. As a consequence, existing product-form results do not apply (see, e.g., \citealt{castro2020matching} for a setting similar to ours but under FCFS). Importantly, our focus is on designing a performant and strategically robust dispatching policy, while this stream of work is mainly concerned with performance evaluation under particular dispatching policies.

\smallskip

\textbf{Queueing games.} Our work also belongs to the literature on managing queues with strategic agents. In his pioneering work, \cite{naor1969regulation} shows that agents' selfish joining behaviors can lead to inefficient system outcomes. 
He argues that to address such inefficiencies, the system provider can rely on monetary transfers to restore first-best outcome. 
Related to our setting, the role of monetary transfers have also been studied in strategic queues with priority. There are two perspectives about this type of queues. First, there are works that analyze priority schemes within queues, see e.g., \cite{kleinrock1967optimum}, \cite{dolan1978incentive},
\cite{hassin1995decentralized}, \cite{afeche2004pricing} and \cite{yang2017trading}. Essentially, in this line of research, customers pay an amount to participate in an auction that determines (or trades) their position (priority) in the queue. The second perspective, which is more aligned with our work, corresponds to priorities between classes of customers (see, e.g., \citealt{cobham1954priority}). 
In this type of model, there is a single server who serves customers with higher priority first (preemptive or non-preemptive regimes can be accommodated), and then continues with customers of lower priorities. In the present work, motivated by regulatory and business constraints outlined in \Cref{sec:introduction}, we do not consider designing monetary incentives. As a consequence, attaining the most efficient outcome might no longer be possible.

In non-monetary settings, there are problems where strategic agents have diverse options when joining a system, and the system provider is constrained to use matching or scheduling policies to optimize the system performance. One of the main challenges in managing these systems is that matching policies alone are, typically, not enough to incentivize agents to make system-wide optimal choices.  \citet{parlakturk2004self}, for example, consider a system with two queues in which strategic agents choose to start their service in one of the queues (and then continue on the other) to minimize their sojourn time.
Their focus is 
on the design of state-dependent scheduling rules
of server resources that have a performance
close to the first-best.
In the context of call centers, \cite{armony2004customer} study a model in which heterogeneous, time-sensitive customers strategically choose between two modes of service (online or callback option). The system's manager 
designs how to allocate servers between the two modes to maximize quality of service subject to certain operational constraints; see \cite{hassin2009equilibrium} for a related framework with two queues. The key features that make our study distinct from these previous works are that in our setting the number of queues is a design choice, not all servers are equal and there is an underlying compatibility graph. Along this line, \citet{caldentey2025designing} recently study service menu design---that is, the design of the underlying compatibility graph---to balance matching rewards and waiting times under the FCFS-ALIS\footnote{The FCFS-ALIS discipline, short for ``first come first served–assign longest idle server,'' operates as follows. When a server becomes available, it is assigned to the longest-waiting customer among those it is compatible with. Conversely, when a customer can be served by multiple idle servers, the customer is matched to the server that has been idle for the longest time.} service discipline, when heterogeneous customers strategically choose which queue to join to maximize expected utility. In contrast, our work focuses on the design of strategically robust, state-dependent matching policies, rather than network design under a fixed FCFS-ALIS discipline.

Finally, our work relates to the literature on allocating heterogeneous items to agents who sequentially accept or reject offers. In the context of kidney exchange, \citet{su2006recipient} study a related market design, but their model admits a queue decomposition because the matching policies are not state-dependent. In contrast, our state-dependent policies intrinsically couple the queues. \citet{su2004patient} examine how queueing disciplines affect patient acceptance and allocation efficiency, showing that LCFS and modified FCFS rules can mitigate rejection. In a ridesharing context, \citet{castro2021randomized} propose randomized FCFS disciplines to reduce cherry-picking in airport queues. %
Unlike these works, we fix the queueing discipline and study the impact of matching policies. Finally, \citet{leshno2019dynamic} study waiting lists with infinitely patient agents and focus on welfare losses from long waits and agent-item mismatch, whereas our setting features short-lived agents and emphasizes throughput maximization.

\smallskip

\textbf{Non-monetary mechanism design.} 
In addition to the queueing literature above, there is also a stream of work 
stemming from scheduling and algorithmic mechanism design 
that considers systems with strategic agents and no monetary transfers. In the prototypical setting, strategic agents 
decide in which  shared resource or machine to complete a task to minimize its finish time (see e.g., \citealt{koutsoupias1999worst}).
The common theme of these works is to analyze the resulting equilibrium and, in some cases, the price of anarchy. For examples, \cite{christodoulou2004coordination} consider the case of coordination mechanisms, \cite{ashlagi2010competing} study  a setting with competing schedulers, \cite{ashlagi2013equilibria} consider a related problem in a queueing setting, and 
\cite{koutsoupias2014scheduling} 
studies a setting in which machines can lie about the time it takes them to complete a task. 
A main difference of our work with the aforementioned papers is that, in our setting, %
tasks have different types which makes them compatible only with a subset of machines (agents).
There are other works, however, that do incorporate a compatibility graph among agents and jobs.
For example, 
\citet{dughmi2010truthful}
consider a version of the generalized assignment problem in which agents can misreport their compatibility with tasks (or, more generally, their value for tasks) 
for a given matching design. 
Motivated by kidney exchange programs,
\citet{ashlagi2015mix} consider the problem of designing a matching mechanism that makes it incentive compatible for the hospital to reveal its compatibility graph. 
Like this past work, we consider the design of matching policies without payments. Unlike this past work, our matching policies critically depend on the dynamically evolving state of the system.

\section{Model and Preliminaries} \label{sec:preliminaries}

We consider a platform that dynamically matches agents and jobs over time. Each agent and each job is associated with a type. We let $\agentset=\{0,1,\dots,\numAgent\}$ and  $\loc = \{0, 1, \ldots, \numLoc\}$ denote the set of agent types and job types, respectively, with $\numAgent=\numLoc$. There is an underlying bipartite compatibility graph between agent types and job types. It has multiple specialized agent types who can perform exactly \emph{one} type of job, and one flexible agent type capable of performing \emph{all} job types. In particular, agents of type \(0\) are \emph{flexible} and can serve any job type in \(\loc\), whereas agents of type \(i \in \{1,2,\cdots,\numAgent\}\) are \emph{specialized} and can serve \emph{only} type \(j=i\) jobs (see \Cref{fig:compatibility_graph_nested}). Type~$0$ jobs can be served \emph{only} by flexible (type~$0$) agents.\footnote{Extensions to the more general setting, where each agent type can serve an arbitrary subset of job types, are discussed in Appendix \ref{appx:fb}.}

This framework is motivated by the examples discussed in \Cref{sec:introduction}. Flexible agents include drivers who can serve any location, handle both passengers and deliveries, or prospective homeowners willing to accept both new and previously owned units. Specialized agents, by contrast, are drivers restricted to a specific location or limited to either passenger or delivery services, or homeowners seeking only new units. The compatibility structure also covers scenarios in which certain jobs (type 0) can be performed only by flexible agents. For example, in a ridesharing setting, job types may correspond to group sizes (e.g., one to six riders) and agent types to vehicle capacities (e.g., UberX or UberXL). Only UberXL drivers---serving as flexible agents in this context---can accommodate groups of five or six riders.

\begin{figure}[htbp]
\centering
    \scalebox{0.65}{\input{figures/compatibility_graph_nested.tikz}}
    \vspace{3mm}\caption{A compatibility graph where there is one flexible agent (type 0) that can serve all job types, and the rest of the agents are specialized and can only serve their corresponding type of jobs.}
\label{fig:compatibility_graph_nested}
\end{figure}

Jobs and agents arrive according to general counting processes, with inter-arrival times that need \emph{not} be independent or identically distributed. Let \( (\lambda_i)_{i \in \agentset} \) denote the long-run average arrival rates for each agent type, and \( (\mu_j)_{j \in \loc} \) the corresponding rates for each job type. Agents have exponentially distributed sojourn times and will abandon the platform at a rate $\renege > 0$ if unmatched. Each agent's type is private information, known only to the agent and unobservable to the platform. However, the detailed arrival processes for all job and agent types are public information, known to both the platform and the agents.

Upon arrival, a job is offered in a uniformly random order, without replacement (effectively instantaneously) to a selected pool of agents (we formally introduce the platform policy in the next section), who either accept and leave the system or decline and continue waiting. However, we assume that jobs have limited patience for rejection; that is, the job will be lost after too many rejections. In particular, after each rejection by an agent, a job survives with probability $p \in [0,1]$.

\subsection{Matching Policy}\label{subsec:matching}

A matching policy uses priority queues to determine how incoming jobs are assigned to available agents.
Without the ability to directly match jobs to specific agent types, the platform instead maintains a set of queues and dispatches jobs to these queues based on a job type-dependent priority list. We allow two (or more) queues to have the same priority for a given job type.  When a job is dispatched to a queue (or multiple queues with the same priority), we assume that it will be assigned uniformly at random among agents in that queue (or those queues). 
Within a queue (or across multiple queues), this corresponds to a lottery, as in the case of affordable housing allocation. In the case of ridesharing, it approximates the closest-dispatch or first-dispatch protocol widely used by ridesharing platforms, which dispatches a rider request to its closest driver (see, e.g., \citealt{castillo2017surge}, \citealt{yan2019dynamic}, and \citealt{besbes2018spatial}).

Upon arrival, agents strategically choose which queue to join (they are not required to join the queue associated with their own type). Agents are told the type of job offered and may choose to decline it, at which point the platform dispatches to another agent. The platform makes offers following the order in the priority list until all queues in the list have been exhausted or the job has been lost. The next section introduces agents' strategies.

More formally, a matching policy $\policy \triangleq (\queueset,\rho)$ consists of a set of queues $\queueset$ and a tuple $\rho = (\rho_0,\dots,\rho_{J})$. The set $\queueset$ can be of any size. For each job $j\in \loc$, $\rho_j$ specifies an ordered list of subsets of $\queueset$ that are mutually exclusive. A type $j$ job would first
be distributed uniformly at random without replacement to the agents in the first subset of queues in the ordered list $\rho_j$ until someone accepts it, the job is lost, or there are no more agents in the subset of queues to offer the job to.

\subsection{Agent Strategies} \label{subsec:agent}

As discussed in \Cref{sec:introduction}, agents may have an incentive to misreport their type. In our model, under a policy $(\queueset,\rho)$, this manifests as a strategic choice of queue. If job offers in a given queue lead to a low expected utility, agents may prefer to avoid that queue. Flexible agents might join a queue intended for specialized agents, while specialized agents might bypass the first queue in the priority list for their compatible job type if too many others are already competing for it. We next describe how agents make their queue-selection choice.

Let $v$ denote the value of securing any compatible job, and $c$ the per-unit-time waiting cost. An agent's expected utility is
\begin{equation*}
v \cdot \mathbb P(\text{match}) - c \cdot \mathbb E[\text{waiting time until a match or abandonment}],
\end{equation*}
where the probability and expectation are taken over the randomness of queue lengths at the time of the agent's arrival,\footnote{Note that the queue length distribution at an agent’s arrival need not coincide with the steady-state queue length distribution, since agent arrival processes may be general. The two distributions coincide when agent arrivals follow a Poisson process, due to the PASTA property~\citep{wolff1982poisson}.} as well as the randomness of future job arrivals, agent arrivals, and agent reneging events. In other words, agents cannot observe the queue lengths upon arrival and must make decisions based on these expectations. We assume $v$ to be constant across job types, interpreting it as an average over any systematic variation in actual pay or job desirability. In ridesharing, this corresponds to platforms adjusting trip prices to equalize expected earnings across trips of different lengths and destinations (see, e.g., \citealt{ma2018spatio}). In affordable housing, this reflects that the gain from securing any acceptable unit is typically large relative to the variation among them, so applicants’ strategic choices are not driven primarily by those differences.

Within the system, all agents face two strategic choices: (1) which queue in \( \queueset \) to join; and (2) whether to accept or decline a job offer upon receiving one of a specific type. We abstract from the decision of whether to join the system, focusing instead on the behavior of agents who are already active participants. For ridesharing, this considers drivers who have already logged in and are deciding how to best use the destination feature or select the type of service they want to provide. In affordable housing, this considers that applicants are already in the system, driven by their housing needs. Given the assumption of constant %
job values, %
it is straightforward that it is a dominant strategy for agents to accept \emph{all compatible} jobs and decline those that are incompatible.

As a consequence, we can represent the strategy of a type \( i \in \agentset \) agent by \( \sigma_i \in [0,1]^{|\queueset|} \), where \( \sigma_{i,q} \) indicates the probability that an agent of type \( i \) joins queue \( q \in \queueset \), with the constraint that \( \sum_{q \in \queueset} \sigma_{i,q} = 1 \). %
Let $\joinprob = (\joinprob_0, \joinprob_1, \dots, \joinprob_\numAgent)$ denote the strategy profile for all agent types. It will be convenient to denote the set of all possible strategy profiles as
$\joinprobset(\queueset) \triangleq \prod_{i\in\agentset}\joinprobset_i(\queueset)$,
where $\joinprobset_i(\queueset)$ is the set of probability measures over the set of queues $\queueset$,  common to all agents $i\in \agentset$.

Each agent chooses the queue that offers the highest expected utility. Consider a strategy profile $\joinprob\in \joinprobset(\queueset)$
under some policy $\policy$. %
Let $W_{i,q}^{\policy}(\joinprob)$ denote the waiting time from arrival until being matched for a type $i \in \agentset$ agent who joins queue $q \in \queueset$, assuming she has \emph{infinite} patience (i.e., never reneges), under strategy profile $\joinprob$, while all other agents in the queues may still renege at rate~$\theta$. This is also termed \emph{virtual waiting time} in the queueing literature \citep{de1985queueing}. 
Let also $R_\theta$ be an exponential random variable with rate $\theta$. The expected utility of an agent of type $i$ who joins queue $q$ under the strategy profile $\joinprob$ is
\begin{equation*}
    u_{i,q}^\policy(\joinprob) \triangleq  v\cdot \PP[W_{i,q}^{\policy}(\joinprob)\le R_\theta] - c\cdot \E[\min\{R_\theta, W_{i,q}^{\policy}(\joinprob)\}],
\end{equation*}
where, again, the probability and expectation are taken over the randomness of queue lengths at the time of the agent's arrival, as well as the randomness of future job arrivals, agent arrivals, and agent reneging events. Note that because $R_\theta$ is an exponential random variable, the probability above equals $\E[e^{-\theta W_{i,q}^{\policy}(\joinprob)}]$ and the expectation is equal to $(1-\E[e^{-\theta W_{i,q}^{\policy}(\joinprob)}])/\theta$. In turn, agents effectively choose the queue that maximizes their probability of being matched before their patience expires, \emph{regardless} of the values of $v$ and $c$.%

\smallskip

We say that $\joinprob$ forms a \emph{Nash equilibrium} if and only if for all $i\in\agentset$ and all $q \in \queueset$: 
\begin{equation} \label{eq:equilibrium} 
    \joinprob_{i,q}>0\quad \: \implies\: \quad u_{i,q}^\policy(\joinprob) \geq u_{i,q'}^\policy(\joinprob),\quad \forall q'\in \queueset. \tag{NE}
\end{equation}
Intuitively, \eqref{eq:equilibrium} requires that if type $i$ agents join queue $q$ with non-zero probability, the expected utility from joining queue $q$ must be the highest among any queue.

The following result guarantees the existence of a Nash equilibrium for any policy and model primitives. %
Its proof consists of writing \eqref{eq:equilibrium} as an equivalent variational inequality. We then cast it as a fixed-point equation, which is guaranteed to have a solution in virtue of Brouwer's fixed-point theorem. 

\begin{proposition}[Existence of Nash Equilibrium] \label{lm:existence_eq}
For any matching policy $\policy$ such that the utility $u_{i,q}^{\policy}(\joinprob)$ is continuous in agents' strategy profile $\joinprob$ for all $i\in\agentset$ and all $q\in\queueset$, 
there exists a strategy profile that forms a Nash equilibrium.%
\end{proposition}

\subsection{Platform's Objective}\label{sec:plt_obj}

Our goal is to investigate how various policies affect the system's performance in terms of {\em throughput}, i.e., the long-run average number of matches. For a policy $\policy$ and strategy profile $\joinprob$, let $M^{\policy}(T;\joinprob)$ denote the total number of matches made up to time $T > 0$, and let $\E[M^{\policy}(T;\joinprob)|\vecAgent^{(0)}=\vecAgent]$ be the expected number of matches up to time $T$ given that the initial number of agents on the platform is $\vecAgent \in \setZ_{\geq 0}^{\numAgent+1}$. The expectation is taken over all the randomness in the system, i.e., the arrival of agents and jobs, agents' abandonment, jobs' chance of loss, and the policy $\policy$. The throughput under a policy $\policy$ and strategy profile $\joinprob$, $\texttt{TP}(\policy;\joinprob)$, is:
\begin{equation}
\label{eq:throughput_defn}
    \texttt{TP}(\policy;\joinprob)=\limsup_{T\rightarrow\infty}\frac{1}{T}\E[M^{\policy}(T;\joinprob)|\vecAgent^{(0)}=\vecAgent].
\end{equation}
Note that this long-run average is invariant to the initial condition $\vecAgent^{(0)}$. Additionally, when the strategy profile is clear from context, we omit its dependence in the notation for the number of matches and throughput, and write $M^{\policy}(T)$ and $\texttt{TP}(\policy)$.

In ridesharing, throughput approximates reliability and growth, a key proxy for both short-term and long-term product success. %
For affordable housing programs like NYC Housing Connect, throughput is also a key measure of success---they aim to place as many eligible applicants into units as possible (subject to fairness and eligibility constraints). %
Across contexts, it captures both short-term efficiency and long-term goals.

\section{Perils of Flexibility Reservation (FR)}
\label{sec:fr}

In this section, we aim to provide insights into the value of accounting for agents' strategic behavior. In a non-strategic setting, our results in this section provide guidelines for how to optimally allocate jobs to agents; while in the strategic setting, we show the perils of maintaining a non-strategic optimal policy. In particular, in \Cref{section:full_info}, we specialize our model to the full information setting, introduce the \emph{flexibility reservation} (FR) policy, and show that it maximizes throughput. Then, in \Cref{sec:private_info_fr_can_be_bad}, we introduce a baseline policy that ignores agents' compatibility---the \emph{random} (RND) policy---and show that flexibility reservation can in fact have substantially poor performance compared to this baseline policy when agents are strategic.

\subsection{Full Information} \label{section:full_info}
We begin by characterizing the throughput-maximizing policy when the platform has \emph{full information} about agent types. In this setting, agents cannot act strategically as the platform can perfectly identify their type. On the side of jobs, whenever there is a compatible agent for a given job, the job can always be matched.  

Intuitively, under full information, flexible agents are more valuable to the platform than specialized agents, as they can serve multiple job types. Accordingly, an effective matching policy should prioritize conserving flexible agents whenever possible. In turn, we consider the \emph{flexibility reservation} (FR) policy, which upon the arrival of a type~\(j\) job, first assigns it to an available specialized agent of the same type, if one exists; otherwise, it dispatches the job to an available flexible (type~\(0\)) agent. Type~\(0\) jobs are always assigned to type~\(0\) agents. \Cref{fig:aggressive-capacity-reservation} illustrates the policy. Formally, the FR policy $\policy^\ACR\triangleq(\FLQQ, \orderedqueue^\ACR)$, offers $\numAgent + 1$ queues, one for each agent type with $\FLQQ = \agentset= \{0,1,2,\cdots,\numAgent\}$, and the ordered list is specified by $\orderedqueue^\ACR_0=(\{0\})$ and $\orderedqueue^\ACR_j=(\{j\},\{0\})$ for all $j\in \loc\setminus\{0\}$. The following result formalizes this discussion. 

\begin{proposition}[First-Best: Flexibility Reservation]
\label{prop:fb}
Under full information about agent types, the flexibility reservation policy maximizes the long-run average throughput. 
\end{proposition}

In \Cref{prop:fb_general} of Appendix~\ref{appx:fb}, we present a stronger version of this statement. In particular, we show that when the platform has full information about agent types, the FR policy maximizes the long-run average throughput \emph{among all state-dependent} matching policies that assign each arriving job to an agent type based on the current system state (i.e., the number of available agents of each type) and the job's type. In the same proposition, we also provide a generalization to arbitrary compatibility graphs.

In practice, types are rarely observable, and agents act strategically. The FR policy characterizes the throughput-maximizing use of flexible capacity when types are known and therefore serves as \emph{first-best} benchmark, and we refer to its throughput as the full information first-best throughput. In ridesharing, it rationalizes dispatch rules that prioritize specialized drivers when available; %
in affordable housing, it corresponds to reserving applicants willing to accept both resale and new units. In the next sections, we quantify losses from private information and the strategic choice of queues, which can be exacerbated by naively applying FR-style policies that ignore strategic behavior. We then design policies inspired by FR that balance flexibility reservation with incentives and deliver high throughput.

\remark The proof of \Cref{prop:fb} relies on a sample path argument up to randomization in reneging. In particular, we fix a sample path of job and agent arrivals (this is why our results hold under general counting processes of agent and job arrivals), but take expectation over agent reneging. The latter is necessary. For example, in a time frame with two arrivals---a flexible agent who abandons quickly and a more patient specialized agent---plus two jobs compatible with the specialized agent, matching the specialized agent first yields a throughput of 1, whereas matching the flexible agent first can yield 2.

\subsection{On the Poor Performance of Flexibility Reservation}\label{sec:private_info_fr_can_be_bad}
The optimality of the FR policy critically relies on the platform's ability to observe agents' true types. As mentioned, however, in practice this information is typically private to the agents.
Consider \Cref{fig:aggressive-capacity-reservation}, which illustrates the scenario under the FR policy described in \Cref{prop:fb}. %
Under this policy, flexible agents may have an incentive to misreport themselves as specialized agents (i.e., $\sigma_{0,0}<1$) to increase their matching likelihood and reduce their waiting times, since jobs are prioritized for assignment to the corresponding specialized agents first. This loss of flexibility can reduce overall system throughput, as fewer flexible agents remain available to serve as a buffer for some jobs that, at certain times, lack corresponding specialized or flexible agents.

To quantify potential throughput losses, we compare the FR policy to a baseline compatibility-agnostic policy that ignores type information. Upon each job arrival, the policy offers the job sequentially uniformly at random (without replacement) to agents in a \emph{single} queue, until a compatible agent accepts. We refer to this as the {\it random} (RND) policy $\pi^{\textrm{RND}} \triangleq (\queueset^{\textrm{RND}}, \rho^{\textrm{RND}})$ which offers a single queue ($\SMQQ= \{0\}$), and the ordered lists contain only queue $0$ for all types of jobs, i.e., $\orderedqueue^\RND_j=(\{0\})$ for all $j \in \loc$. %
We use it to assess the value---and potential downsides---of policies that attempt to exploit type information; additional benchmark policies are studied in \Cref{sec:sim_experiments}. 

Operationally, the RND policy eliminates both type-based prioritization and agents’ ability to (mis)report types: there is a single queue and jobs are offered at random within the waiting pool. In ridesharing, this corresponds to not offering a driver-destination feature and not allowing drivers to pre-declare whether they accept rides, deliveries, or both, yielding simply the closest-dispatch protocol that can be well approximated with random under uniform spatial demand. 
In affordable housing, lottery mechanisms are analogous: resale and new units are offered via lotteries among all applicants who expressed compatibility with those units (effectively a random allocation in equilibrium cf. \Cref{sec:introduction}). In both cases, the random policy provides a natural baseline against which we can quantify how private information and strategic queue choice affect throughput.

It is immediate that $\policy^{\textrm{RND}}$ is both strategy-proof and strategy-independent, since $\sigma_{i,0}=1$ for all $i \in \agentset$ is the only feasible strategy when there is a single queue.
Our next result shows that there exists an instance such that, at equilibrium, $\policy^\ACR$ can achieve \emph{substantially worse} throughput than $\pi^{\textrm{RND}}$.%

\begin{figure}
\centering
  \resizebox{0.96\textwidth}{!}{%
    \input{figures/MQ.tikz}
    }
  \caption{The Flexibility Reservation (FR) Policy. %
  \normalfont{The solid arrows from jobs to queues represent the first priority, and the dash-dotted arrows represent the second priority. For each specialized agent $j \in \{1, \ldots, J\}$, it is clear that joining their corresponding queue ($\sigma_{j,j} = 1$) is a weakly dominant strategy. Therefore, we do not depict edges corresponding to other strategies in the figure.}}
  \label{fig:aggressive-capacity-reservation}

\end{figure}

\begin{proposition}[FR Can Be Worse Than RND] \label{thm:flip1-alt}
For any $\delta>0$, there exists an instance and an equilibrium strategy $\sigma^{\textrm{FR}}$ under the FR policy such that 
\[
\normalfont{\texttt{TP}}(\FLQ; \sigma^{\textrm{FR}}) < \left(\tfrac{1}{2}+\delta\right)\, \normalfont{\texttt{TP}}(\pi^{\textrm{RND}}).
\]
\end{proposition}

The FR policy can perform poorly when flexible agents find a particular specialized queue especially attractive, e.g., because those jobs arrive frequently or those specialized agents are scarce. Reserving flexible capacity helps only when the other specialized jobs and the flexible jobs also arrive at high rates. This creates a tension: misreporting hurts throughput the most when the rest of the system is busy, yet in that regime, the incentive to misreport can be weaker.

To prove the proposition, we exhibit an instance where this tension resolves against FR. The instance has only flexible agents, two job types, and two main characteristics: (i) many flexible jobs that arrive one step later, so agents must wait; and (ii) fewer specialized jobs that are available immediately upon agents’ arrival.
Because agents must wait for flexible jobs, both options are similarly attractive, and under FR, agents are willing to join the specialized queue.
The result is that FR matches all specialized jobs but sacrifices all flexible jobs. By contrast, the compatibility-agnostic random policy still matches all specialized jobs and, in addition, it captures an extra share of flexible jobs, matching them to agents who are still around when the next batch arrives. As we show below, in an instance with these characteristics, FR can deliver as little as about half the throughput of random.

\smallskip

\paragraph{Proof of \Cref{thm:flip1-alt}.}
Consider the compatibility graph in \Cref{fig:aggressive-capacity-reservation}, restricted to two job types $\loc=\{0,1\}$ with $\lambda_1=0$ and per-unit-time waiting cost $c=0$.  Suppose type~0 agents arrive in deterministic batches of size $D$ at epochs $0,T,2T,\dots$, with interarrival time $T>0$. Likewise, type~0 jobs arrive in identical batches of size $D$ at epochs $T^-,2T^-,\dots$, i.e., \emph{just before} the arrival of the corresponding batch of type~0 agents. In this setup, to get a type $0$ job, each type~0 agent must wait until the next batch of jobs arrives, yielding a constant waiting time of $T$ to get a type $0$ job. Hence, the probability that any type~0 agent remains in queue $0$ until matched with a type $0$ job is $e^{-\theta T}$.

Now, type~1 jobs arrive deterministically in batches of size $e^{-\theta T}D$ (assumed to be integral) at epochs $0^+,T^+,2T^+,\dots$, i.e., \emph{just after} the arrival of type~0 agents. Thus, if \emph{all} type~0 agents switch to queue~1, they again achieve matching probability $e^{-\theta T}$,\footnote{To see this formally, let $C$ be the steady-state number of carry-over agents, then $(C + D - e^{-\theta T}D)e^{-\theta T} = C$. Additionally, the probability $q$ of being matched in queue 1 under FR satisfies $q= \frac{e^{-\theta T}D}{C + D} + \left(1-\frac{e^{-\theta T}D}{C + D}\right)e^{-\theta T} q$.  } identical to staying in queue~0. Therefore, $\sigma^{\textrm{FR}}_{0,1}=1$ is an equilibrium strategy profile. In this case, the system throughput under  $\sigma^{\mathrm{FR}}$ equals
\[
\texttt{TP}(\FLQ;\sigma^{\mathrm{FR}}) \;=\; \frac{e^{-\theta T}D}{T}.
\]

By contrast, under the random (RND) policy, throughput has two contributions:  
(i) immediate matches between type~0 agents and type~1 jobs upon their arrival, yielding $e^{-\theta T}D/T$, and  
(ii) deferred matches of the surviving type $0$ agents from the current batch with the next batch of type~0 jobs, yielding $e^{-\theta T}(1-e^{-\theta T})D/T$.  
Thus,
\[
\texttt{TP}(\pi^{\mathrm{RND}}) \;=\; \frac{e^{-\theta T}D}{T} \;+\; \frac{e^{-\theta T}(1-e^{-\theta T})D}{T}.
\]

It follows that
\[
\frac{\texttt{TP}(\FLQ;\sigma^{\mathrm{FR}})}{\texttt{TP}(\pi^{\mathrm{RND}})}
=\frac{1}{2-e^{-\theta T}}.
\]
As $T\to\infty$ (and $D\to\infty$ to keep $e^{-\theta T}D\in\mathbb Z_{>0}$), $e^{-\theta T}\to 0$, and the ratio converges to $1/2$. %
$\square$

\section{Flexibility Reservation with Fallback (FRfb)} \label{sec:flq_with_recourse}

The possibility that the FR policy may yield arbitrarily lower throughput than a random policy is both intriguing and somewhat unsatisfactory. In the context of ridesharing, for example, this suggests that platforms might sometimes achieve better performance by disregarding drivers' stated destination preferences when making matching decisions---that is, by not offering the driver destination feature at all, as was the practice at Uber and Lyft prior to 2015.

In this section, we start in \Cref{sec:frfb_good} by introducing the \textit{Flexibility Reservation with Fallback} (FRfb) policy, which addresses the shortcomings of the flexibility reservation policy by potentially offering jobs to agents in seemingly incompatible queues. We then show that while this policy may incentivize more misreporting from flexible agents---a form of Braess paradox---it is never worse than the compatibility-agnostic random policy, and is strictly better in most problem instances. At the end of the section, %
we illustrate our results and insights with a numerical example. %

\subsection{FRfb: Policy Definition and Performance Guarantees}\label{sec:frfb_good}

\begin{definition}[Flexibility Reservation with Fallback]\label{dfn:FR_fb}
The \emph{Flexibility Reservation with Fallback} (FRfb) policy, denoted as $\policy^\RCR$, sets $\queueset^\RCR=\agentset$ queues, and it specifies priorities by  $\rho^\RCR_j = (\{j\}, \{0\}, \agentset \setminus \{j, 0\})$ for all $j \in \{1,2,\cdots,\numLoc\}$ and $\rho^\RCR_0=(\{0\}, \agentset \setminus \{0\})$ for type $0$ jobs.

\end{definition}

\smallskip

Similar to the FR policy, the FRfb policy also prioritizes reserving flexibility. It first attempts to match a job of type~$j$ with its corresponding specialized agents, and, if unavailable, then with flexible agents of type~$0$. A key distinction of the FRfb policy, however, is that if a match is not found after exhausting all agents in the compatible queues, it uses a fallback option: the job will be dispatched uniformly at random to agents in the \emph{incompatible} queues until it is either matched or lost. This approach leverages the possibility that some compatible agents may be ``hidden'' within the incompatible queues. This definition can also be extended to general compatibility graphs (see \Cref{dfn:FR_fb_general} in Appendix \ref{appx:frfb}). We now provide an example to illustrate the construction of the FRfb policy.

\smallskip

\begin{example}
Consider the case where $\agentset=\loc = \{0,1\}$, with type $0$ agents being flexible and type $1$ agents being specialized. The FRfb policy is specified as $\rho^\RCR_0 = (\{0\},\{1\})$ and $\rho^\RCR_1 = (\{1\},\{0\})$. Compared to the original FR policy, the key difference of the FRfb policy is the addition of \emph{seemingly incompatible arrows} from type $0$ jobs to queue $1$ (see the blue dash-dotted arrow in \Cref{fig:robust_capacity_reservation}) as a fallback option. If agents are not strategic and always join their designated queues, this mechanism does not provide any additional benefit. However, when flexible agents act strategically, some may choose to join queue $1$. As a result, type $0$ jobs that would otherwise have been lost can instead be completed by these flexible agents who have joined a specialized queue. 
The same principle applies in settings with additional specialized types. If many flexible agents join a particular specialized queue, under FR, this can cause other specialized jobs to go unserved. However, FRfb provides the option to match those jobs to the ``hidden" flexible agents.
\end{example}

\begin{figure}[htbp]
    \centering
    \scalebox{0.6}{\input{figures/braess.tikz}}
    \caption{An example of FRfb policy with two types of jobs. 
    \textnormal{
    Black solid arrows from jobs to queues represent first-priority matches, while dash-dotted arrows indicate second-priority matches; the blue dash-dotted arrows highlight the added fallback option.}}
    \label{fig:robust_capacity_reservation}
\end{figure}

We next provide insights into the FRfb policy by first examining some of its equilibrium properties and then benchmarking its performance.
At first glance, the FRfb policy makes better use of flexible agents: jobs that cannot be fulfilled by agents in the compatible queues will be dispatched to the rest of the agents, and can potentially be fulfilled by the compatible agents therein. This benefit, however, does come at a cost since the resulting higher service rate for the incompatible queues may incentivize more flexible agents to pretend to be less flexible. The following result formalizes this observation with an example.

\begin{proposition} \label{prop:two_type_more_switch}
Consider the setting illustrated in \Cref{fig:robust_capacity_reservation} with two types of jobs $\loc=\{0,1\}$ and suppose that the utility $u_{i,q}^{\policy}(\joinprob)$ is continuous in agents' strategy profile $\joinprob$ for all $i\in\agentset$ and all $q\in\queueset$. For any arrival processes of agents and jobs and any model parameters, %
there exist equilibrium strategy profiles $\joinprobFLQR$ and $\joinprobFLQ$ such that the fraction of flexible type~$0$ agents joining queue~$1$ in equilibrium is weakly higher under the FR policy than under the FRfb policy, i.e., $\joinprobFLQR_{0,1} \ge \joinprobFLQ_{0,1}$.
\end{proposition}

\Cref{prop:two_type_more_switch} is related to Braess' paradox (see, e.g., \citealt{braess1968paradoxon}): adding a seemingly beneficial edge can, in equilibrium, result in an undesirable outcome. 
In our setting, this proposition shows that by adding an additional edge,  more flexible agents than under FR may choose to join queue $1$. As a consequence, the platform may possibly end up using flexible agents in an inefficient manner by matching more specialized type $1$ jobs to flexible agents.

Fortunately, despite this potential drawback of the FRfb policy, 
it delivers performance that sharply contrasts with that of the FR policy (cf. \Cref{thm:flip1-alt}).
While the FR policy can perform significantly worse in equilibrium compared to the RND policy, %
we show in the following theorem that the FRfb policy's performance is \emph{always} better than the randomized policy, in equilibrium or not.

\begin{theorem}[\normalfont{\textsc{Robust Performance of the FRfb Policy}}]
\label{thm:robust_improvement_FLQR} 
For any strategy profile %
$\joinprob\in\joinprobset(\queueset)$ such that 
$\joinprob_{i,i}=1,\:\forall i\in\{1,2,\cdots, \numLoc\}$, 
the throughput achieved by the FRfb policy is always higher than or equal to that under the RND policy:
\begin{align*}
    \normalfont{\texttt{TP}}(\FLQR;\joinprob) \ge \normalfont{\texttt{TP}}(\pi^{\textrm{RND}}).
\end{align*}
Moreover, under the FRfb policy, it is a weakly dominant strategy for specialized agents to stay in their corresponding queues, i.e., $\joinprob_{i,i}=1,\:\forall i\in\{1,2,\cdots, \numLoc\}$. 
\end{theorem}

We make two important observations. First,
\Cref{thm:robust_improvement_FLQR} shows that the FRfb policy outperforms the RND policy under \emph{any} strategy chosen by the flexible agent type. %
This guarantee is particularly strong, as it relies on \emph{no} equilibrium assumption beyond the very mild requirement that specialized agents play weakly dominant strategies. 
Second, note that under the compatibility graph in \Cref{fig:robust_capacity_reservation}, the FRfb policy coincides with the RND policy when all flexible agents switch to queue~$1$. However, as the switching fraction decreases, FRfb begins to outperform RND. When there are more than two job types, it is noteworthy that there is a \emph{strict} separation between FRfb and RND---FRfb is strictly better, since no matter how flexible agents pretend to be specialized, the two policies can never coincide.

In Appendix~\ref{appx:frfb}, we provide a generalized definition of the FRfb policy applicable to arbitrary compatibility graphs, where each agent type may serve an arbitrary subset of job types. Then, in \Cref{sec:sim_experiments}, we demonstrate that FRfb again achieves strong performance across a variety of compatibility graphs and problem instances, even when compared against an approximate second-best benchmark.

The proof of \Cref{thm:robust_improvement_FLQR} hinges on a key observation: for any type~$j \neq 0$ job, and for any fixed numbers of agents of each type on the platform, the probability of matching this job to a flexible type~$0$ agent under FRfb is always (weakly) lower than under the random policy. This holds \emph{regardless} of how many flexible agents pretend to be specialized by joining the specialized queues in FRfb. 
To see this, suppose that at time~$t$, the numbers of agents of each type on the platform are given by $\vecAgent \supt = (\agent\supt_0, \agent\supt_1, \dots, \agent\supt_\numLoc)$. Conditional on a type $j$ job being matched, the RND policy assigns the job to a flexible agent with probability $\agent_0\supt / (\agent_0\supt + \agent_j\supt)$. In contrast, the FRfb policy randomizes between the $\agent_j\supt$ type~$j$ agents and those flexible agents who hide in queue~$j \neq 0$, whose number is at most $\agent_0\supt$. This shows that FRfb better preserves flexible agents than the RND policy. The proof uses a sample path argument, taking the expectation over reneging and jobs' patience (cf. Remark 1), to formalize this intuition.

To conclude, in addition to the performance guarantee in \Cref{thm:robust_improvement_FLQR}, \Cref{thm:gap_first_best} below establishes a $1/2$ throughput guarantee for FRfb
relative to the full-information first-best throughput. This guarantee applies to \emph{any} compatibility graph beyond \Cref{fig:compatibility_graph_nested} (as defined in Appendix \ref{appx:fb}) and \emph{any} strategy profile, assuming jobs have infinite patience for agent rejections. %

\begin{proposition}[Half Approximation] \label{thm:gap_first_best}
When $p=1$, the FRfb policy achieves at least half of the full-information first-best throughput under any strategy profile $\joinprob \in \joinprobset(\queueset)$.
\end{proposition}
We note that the proof of the proposition implies that this $1/2$ guarantee holds for any non-idling matching policy under the given strategy profile, i.e., a job is always matched whenever a compatible agent is available. In turn, both FRfb and RND enjoy the guarantee. We view this result as completing the picture about
the FRfb: it is always better than RND and a 1/2 approximation of the first best. Figure 5 complements this picture, showing consistent high throughput for FRfb.

\smallskip

\subsection{Illustration of Results and Insights}\label{sec:ill-res-ins}
In this subsection, we use a numerical example to illustrate our theoretical results. We consider the compatibility graph shown in \Cref{fig:compatibility_graph_nested} with two types of agents and jobs $\agentset=\loc = \{0,1\}$ (type $0$ agent is flexible and type $1$ agent is specialized). %
Informally---details follow---type~0 jobs have a high arrival rate only during certain periods and a low rate most of the time, while Type~1 jobs arrive at a more consistent rate. 
A real-world analogy is Uber drivers choosing between food delivery and ridesharing: UberEats requests spike at specific times (e.g., dinner), whereas Uber rides requests are comparatively steady. 

Under the FR policy, flexible agents may choose to serve only the more consistent job type, capping the system’s throughput at that rate. In contrast, the RND policy assigns all jobs to the full pool of agents, allowing type~0 jobs to be matched during high-arrival periods and potentially achieving a substantial throughput advantage. FRfb can outperform both policies by reserving some flexible capacity and offering flexible jobs to the pool of seemingly incompatible agents.

\paragraph{Instance construction.} Formally, type~$0$ and type~$1$ agents arrive according to independent Poisson processes with rates $\agentarr_0$ and $\agentarr_1$, respectively; type~1 jobs arrive according to a Poisson process with rate $\jobarr_1$, while type~0 jobs follow a Markov-modulated process with states $\{L,H\}$---a high state~$H$ with rate $\mu_{1,H}$ and a low state~$L$ with rate $\mu_{1,L}$. Let $\kappa_{L\rightarrow H}$ and $\kappa_{H\rightarrow L}$ denote the transition rates between the two states. We choose parameters so that type~0 job arrivals become increasingly spiky, arriving in bursts---see the caption of \Cref{fig:robust-per} for the specific parameters. %
We consider two scenarios, one in which jobs have infinite patience, $p=1.0$, %
and another in which jobs will be lost  if declined by on average five agents, $p=0.8$. 
We note that ridesharing drivers typically have at most $15$ seconds to accept or decline a job, so this gives a rider patience for about one minute, which is quite conservative. %

\paragraph{Policies.} We compare three policies: (1) the FR policy whose performance is independent of the value of $p$ as it never sends jobs to incompatible agents; (2) the RND policy that sends jobs to agents uniformly at random whose performance depends on $p$; (3) our proposed FRfb policy whose performance also depends on $p$ as type $0$ jobs will be sent to possibly type $1$ agents if there is no agent in queue $0$.

\paragraph{Discussion of results.} We report the policies' performance in \Cref{fig:robust-per} as the arrivals of type~0 jobs become increasingly spiky ($\varepsilon$ decreases).  \Cref{fig:7a} presents throughput as a fraction of the first-best---the maximum throughput under full information---while \Cref{fig:7b} shows the equilibrium fraction of flexible agents who pretend to be specialized (i.e., \(\sigma_{0,1}\) in equilibrium). 
\Cref{fig:7a} shows that the performance of the FR policy (red curve) continues to deteriorate as $\varepsilon$ decreases: queue 1 becomes more appealing so that more and more flexible agents join it (see \Cref{fig:7b}). As this happens, we increasingly lose the opportunity to match type 0 jobs, deteriorating the performance of FR, and making it even worse than random (cf. \Cref{thm:flip1-alt}). 
In contrast, the proposed FRfb policy consistently delivers near-optimal throughput. %
FRfb is slightly outperformed by FR under $p=0.8$ when $\varepsilon$ is relatively large, reflecting the Braess' paradox result (the comparison of the equilibrium fraction in \Cref{fig:7b} also corroborates \Cref{prop:two_type_more_switch}). %
It is interesting to observe that the throughput of the FRfb policy does \emph{not} necessarily decrease as the survival probability $p$ decreases. On the one hand, a low value of $p$ makes type $0$ jobs harder to match successfully; on the other hand, it also makes queue $1$ less attractive as it effectively increases the waiting time of flexible agents in it, resulting in fewer flexible agents switching to queue 1 (see \Cref{fig:7b}). In fact, in the extreme case of $p=0$---a job is immediately lost after one agent rejection---FR and FRfb become the same in equilibrium. Finally, note that the throughput comparison between the RND and FRfb policies in \Cref{fig:7a} under the same values of $p$ confirms \Cref{thm:robust_improvement_FLQR}. %

\begin{figure}[htbp]  
\captionsetup[subfigure]{justification=centering}
\centering 

 \begin{subfigure}[b]{0.48\textwidth}
 \centering
    \begin{tikzpicture}
    
    \begin{axis}[
    width=3in,
    legend style = { at = {(0.6,0.75)}},
    no marks,
    xlabel = $\log_{10}(1/\varepsilon)$,
    xmax   = 4,
    xmin   = 0,
    xtick={0,0.5,1,1.5,2,2.5,3,3.5,4},
    ymax   = 1.01,
    ymin   = 0.65,
    label style={font=\small},
    tick label style={font=\small},
    legend style={
        at={(0.45,0.5)}, %
        anchor=north,     %
        font=\scriptsize
    }
    ]
    \addplot[smooth, color = red,
    line width=1.2pt,
    solid] table[x=eps_exponent, y = FR]{figures/data/new_utility.txt};
    \addlegendentry{FR};
    \addplot[smooth,color = green,
    line width=1.2pt, dashed] table[x=eps_exponent, y = RND_08]{figures/data/new_utility.txt};
    \addlegendentry{RND (p=0.8)} ;
    \addplot[smooth,color = green,
    line width=1.2pt] table[x=eps_exponent, y = RND_10]{figures/data/new_utility.txt};
    \addlegendentry{RND (p=1.0)} ;
    \addplot[smooth, color = blue,
    line width=1.2pt, dashed] table[x=eps_exponent, y = FRfb_08]{figures/data/new_utility.txt};
    \addlegendentry{FRfb (p=0.8)};
    \addplot[smooth, color = blue,
    line width=1.2pt] table[x=eps_exponent, y = FRfb_10]{figures/data/new_utility.txt};
    \addlegendentry{FRfb (p=1.0)};
    \end{axis}
  \end{tikzpicture}%
    \caption{Throughput (as a fraction of FB)}
  \label{fig:7a}
\end{subfigure}
~
\begin{subfigure}[b]{0.48\textwidth}
\centering
    \begin{tikzpicture}
    \begin{axis}[
    width=3in,
    legend style = { at = {(0.6,0.75)}},
    no marks,
    xlabel = $\log_{10}(1/\varepsilon)$,
    xmax   = 4,
    xmin   = 0,
    xtick={0,0.5,1,1.5,2,2.5,3,3.5,4},
    ymax   = 1.05,
    ymin   = -0.05,
    label style={font=\small},
    tick label style={font=\small},
    legend style={
        at={(0.73,0.92)}, %
        anchor=north,     %
        font=\scriptsize
    }
    ]
    \addplot[color = red,
    line width=1.2pt] table[x=eps_exponent, y = eq_FR]{figures/data/new_utility.txt};
    \addlegendentry{FR} ;

    \addplot[color = blue,
    line width=1.2pt, dashed] table[x=eps_exponent, y = eq_FRfb_08]{figures/data/new_utility.txt};
    \addlegendentry{FRfb (p=0.8)};
    \addplot[color = blue,
    line width=1.2pt] table[x=eps_exponent, y = eq_FRfb_10]{figures/data/new_utility.txt};
    \addlegendentry{FRfb (p=1.0)};
    
    \end{axis}
  \end{tikzpicture}
  \caption{Equilibrium prob. of joining queue $1$}
  \label{fig:7b}
\end{subfigure}
\caption{Performance of various policies. {\normalfont \Cref{fig:7a} depicts throughput as a fraction of the first-best throughput for policies FR, RND, and FRfb. \Cref{fig:7b} shows the fraction of flexible agents that join queue $1$ in equilibrium for policies FR and FRfb. The model primitives are $\lambda_1=30,\:\mu_1 = \log^2_{10}(1/\varepsilon) ,\:\lambda_0 = \log^2_{10}(1/\varepsilon)/\varepsilon,\:\mu_{0,L}=2, \mu_{0,H} = \log_{10}(1/\varepsilon)/\varepsilon^2, \kappa_{L\to H}=\varepsilon, \kappa_{H\to L}=1$ and $\theta = \log_{10}(1/\varepsilon)/\varepsilon$.
We consider two values of job survival probabilities $p\in\{0.8, 1.0\}$ after each agent rejection.
}}
\label{fig:robust-per}
\end{figure}

\section{Simulation Experiments} \label{sec:sim_experiments}

Complementing our previous theoretical results, in this section, we use simulation to compare four matching policies %
across three increasingly general families of compatibility graphs. For each policy–graph pair, we compute the induced equilibrium joining behavior via a replicator-dynamics procedure and then simulate the resulting system. We report throughput, jobs lost due to excessive rejections, the share of agents who misreport, and average match probabilities.

\subsection{Simulation Environments and Parameters}\label{sec:env-params}

\paragraph{Graph families.} For each agent type $i\in\agentset$, let $\loc(i)\subset\loc$ be the subset of job types this agent type can serve. %
We consider three types of compatibility graphs: (1) $G_1$ is the compatibility graph depicted in \Cref{fig:compatibility_graph_nested} with $\agentset=\loc=\{0,1,\cdots,\numLoc\}$ where the fully flexible type $0$ agent can serve all types of jobs ($\loc(0)=\loc$) and type $i\neq 0$ agents are specialized to serve only type $i$ jobs ($\loc(i)=\{i\},\:\forall i\neq 0$); (2) $G_2$ is another compatibility graph with $\agentset=\loc=\{0,1,\cdots,\numLoc\}$  and $\loc(i)=\{i, i+1,\cdots, \numLoc\}$, i.e., all agent types are completely nested with $\loc(0)\supset \loc(1)\supset\cdots\supset\loc(\numLoc)$; and finally (3) $G_3$ is a complete graph with $|\agentset| = |2^\loc|-1$; that is, for each nonempty subset of job types in $\loc$, there exists a compatible agent type.

The graphs $G_2$ and $G_3$ generalize our baseline compatibility graph $G_1$. The monotone structure of $G_2$ captures hierarchical compatibility; for example, in ridesharing, higher-capacity vehicles (e.g., minivans) can serve everything medium- and lower-capacity vehicles (e.g., SUVs and hatchbacks) can. By contrast, $G_3$ imposes no structure and allows arbitrary compatibility patterns.
The generality of $G_3$ allows us to effectively simulate a wide range of different compatibility configurations by changing the arrival rates of agents and jobs.  

\paragraph{Model primitives.}
We set the number of job types $\numLoc=4$. %
We assume that agents and jobs arrive to the platform according to Poisson processes and the total arrival rates of jobs and agents are $\sum_{j\in\loc}\jobarr_j = 4\numLoc$ and $\sum_{i\in\agentset}\agentarr_i = 5\numLoc$, respectively. 
The total arrival rate of agents is made slightly higher than the rate of jobs to balance overall demand and supply as agents renege in the system. The reneging rate of all agents is assumed to be $\theta = 1$. Given the total arrival rates of agents, the proportion of each type is sampled from a %
Dirichlet distribution with parameter $\alpha = (\alpha_i)_{i\in\agentset}$. We consider two types of distributions. One is a symmetric Dirichlet distribution with $\alpha_i=1,\:\forall i\in\agentset$, i.e., on average each type of agent has the same proportion; the other one has $\alpha_i = |\loc(i)|,\:\forall i\in\agentset$, i.e., on average type $i$ agent has a proportion of $|\loc(i)|/(\sum_{i\in\agentset}|\loc(i)|)$---more flexible agents on average have a larger presence. On the other hand, the proportion of each job type is always sampled from a symmetric Dirichlet distribution with all parameters equal to one. Finally, similarly to the example in \Cref{fig:robust-per}, we assume that each rejection by an agent results in an independent event in which the job survives with probability $p\in\{0.8, 1.0\}$. 

For each distribution choice, we average the performance over $100$ random draws of $(\agentarr_i)_{i\in\agentset}$ and $(\mu_j)_{j\in\loc}$.  The length of each simulation run is $T=1,000$. We introduce the policies tested and explain how the agents' equilibrium strategy profile is computed within our simulation in the next subsection. 

\subsection{Policies Evaluated and Equilibrium Computation}\label{sec:pol-eval}
In what follows, we introduce the four different policies we simulate---FR, RND, FRfb, and an iterative approximation to the second-best policy---and an upper bound on the first-best  throughput with full information. %
We also detail an iterative algorithm that can be used to compute agents' equilibrium strategy profile under any policy and graph structure.

\begin{enumerate}
\item \textbf{A fluid upper bound on the first-best policy with full information (UB).} Computing the exact first-best matching policy for a general compatibility graph (e.g., $G_3$) is intractable due to the curse of dimensionality, since agent flexibility levels may not admit a strict ordering for a given job type (e.g., one agent type may serve job types~1 and~2, and another may serve job types~1 and~3, making it unclear which type is more flexible when assigning job type~1). In Appendix \ref{appx:fb}, we present a tractable upper bound on the optimal throughput with full information by solving a fluid linear programming relaxation of the dynamic stochastic matching problem. %

\item \textbf{The flexibility reservation policy (FR).} 
The FR policies for compatibility graphs \(G_1\) and \(G_2\) admit explicit characterizations, with the policy for \(G_2\) extending the FR policy in \Cref{prop:fb} (see Appendix \ref{appx:fb}). These policies are also first-best for \(G_1\) and \(G_2\) (see \Cref{prop:fb_general} in Appendix \ref{appx:fb}). For the complete compatibility graph \(G_3\), computing the exact optimal policy is intractable; instead, we construct a heuristic policy derived from the optimal dual values for each agent queue in the same fluid linear program in Appendix \ref{appx:fb} used for the upper bound. These dual values represent the shadow prices of an agent in each queue and induce, for each job type \(j \in \loc\), a strict priority ordering of compatible agent queues. This ordering is determined by the magnitude of the dual values, replacing the flexibility-level ranking in \Cref{prop:fb}, which is not well-defined in \(G_3\). See Appendix \ref{appx:fb} for details.

\smallskip

\item \textbf{Random policy (RND).} Jobs are sent to all agents uniformly at random without replacement. 

\smallskip

\item \textbf{Flexibility reservation with fallback policy (FRfb).} For \(G_1\), this corresponds to the FRfb policy as given by \Cref{dfn:FR_fb}. 
For \(G_2\) and \(G_3\), we extend the FRfb policy as given by \Cref{dfn:FR_fb_general} of Appendix~\ref{appx:frfb}, 
to accommodate general compatibility graphs.

\smallskip

\item \textbf{An iterative procedure to approximate a Stackelberg equilibrium/second-best policy (SB).} %
We adopt a popular heuristic (see, e.g., \citealt{marcotte1992efficient}) to approximate the solution of a Stackelberg game---in our setting, the leader is the platform, which sets a matching policy, and the followers are agents reporting their types. This yields the so-called \emph{second-best} policy, i.e., the throughput-maximizing matching policy when agents behave strategically. The platform maintains a queue for each agent type. The iterative procedure starts by assuming that all agents truthfully report their types by joining their designated queues, after which the matching policy is optimized based on the reported type rates. Agents then update their joining strategies to form a new equilibrium, and the platform re-optimizes the policy using the new reported rates. This process repeats until neither the platform nor the agents have an incentive to deviate. The converged policy is generally not a Stackelberg (subgame-perfect) equilibrium but rather a Nash equilibrium in which neither party can unilaterally improve. We report the performance of the policy that achieves the \emph{highest} equilibrium throughput during this process. For the complete compatibility graph \(G_3\), policy updates are obtained by \emph{resolving} the fluid linear program in Appendix~\ref{appx:fb} with the current reported type rates to compute updated dual values, which are then used to construct the policy for \(G_3\).
For $G_1$ and $G_2$, the optimal policy \emph{does not} depend on the reported type rates (see \Cref{prop:fb_general} in Appendix \ref{appx:fb}), so the procedure converges in a single iteration, yielding the same performance as the FR policy.

\end{enumerate} 

\smallskip

\paragraph{Computing equilibrium joining strategy.} Given a policy, computing the joining equilibrium exactly with a large number of agent types can be a daunting and computationally intensive task. Instead, we use an evolutionary dynamics technique called replicator dynamics to 
approximate the agents' equilibrium strategies. These dynamics are initialized with a strategy profile $\sigma^0 = (\sigma_i^0)_{i\in\agentset}$, e.g., a random joining strategy profile. 
Without loss of generality (cf. \Cref{subsec:agent}), we set 
$u^\policy_{i,q}(\sigma) = \PP[W_{i,q}^{\policy}(\joinprob)\le R_{\theta}]$---the payoff of an agent of type $i$ joining queue $q$ under policy $\policy$ if all other agents adopt the strategy profile $\joinprob$.
Let $\bar{u}_i^\policy(\joinprob) = \sum_{q\in\mathcal{Q}} \joinprob_{i,q} u^\policy_{i,q}(\joinprob)$ be the average payoff of a type $i$ agent. We use the following updating rule:
\be
\label{eq:replicator_dynamics}
\joinprob^{t+1}_{i,q} - \joinprob^{t}_{i,q} = \joinprob^{t}_{i,q}\cdot\frac{ u^\policy_{i,q}(\joinprob^t) - \bar{u}^\policy_{i}(\joinprob^t) }{\bar{u}^\policy_{i}(\joinprob^t)},\quad \forall i\in\agentset,\: q\in\mathcal{Q}.
\ee
It is not hard to see that as long as $\sum_{q\in\mathcal{Q}}\joinprob_{i,q}^t = 1$, we have $\sum_{q\in\mathcal{Q}}\joinprob_{i,q}^{t+1} = 1$, i.e., $\joinprob^{t+1}$ remains a valid strategy profile. Moreover, if the replicator dynamics \eqref{eq:replicator_dynamics} has a stationary point $\joinprob^\ast$, i.e., $\joinprob^{t+1} = \joinprob^{t} = \joinprob^\ast$ and $\joinprob^\ast$ is asymptotically stable in the sense that there exists a neighborhood of $\joinprob^\ast$ such
that starting from any $\joinprob^0$ in this neighborhood, the replicator dynamics \eqref{eq:replicator_dynamics} approaches $\joinprob^\ast$,  %
then $\joinprob^\ast$ is a Nash equilibrium strategy \citep{taylor1978evolutionary, cressman2003evolutionary}. For each iteration $t$, we estimate $u^\policy_{i,q}(\joinprob^t)$ using Monte-Carlo simulation (empirical counts of matching frequency) and terminate the dynamics as long as the weighted variance of agent payoff $\sum_{q\in\mathcal{Q}}\joinprob^t_{i,q}(u^\policy_{i,q}(\joinprob^t) - \bar{u}^\policy_{i}(\joinprob^t))^2$ is small enough for each agent type $i\in\agentset$.

For each matching policy, we compute the equilibrium joining strategy using the aforementioned replicator dynamics with Monte-Carlo simulation. After each run, the payoff (empirical counts of matching frequency) is calculated and the strategy profile is updated according to \cref{eq:replicator_dynamics}. This procedure continues until convergence.

\smallskip

\subsection{Results}\label{sec:sim_results}

\Cref{tb:sim_results} below provides the computational results. The first four rows report the throughput of different policies as a fraction of the throughput upper bound obtained by solving the fluid linear program. 
It can be observed that the FRfb policy consistently delivers the highest throughput. %
It is also interesting to see that the FRfb policy under rejection penalty ($p=0.8$) outperforms the RND policy without any rejection penalty ($p=1.0$), strengthening the message of \Cref{thm:robust_improvement_FLQR}. The performance of the FR policy degrades especially when there is a larger proportion of flexible agents in the system (the case of $\alpha_i=|\loc(i)|,\:\forall i\in\agentset$) as the incentive of under-reporting their types increases. In these cases, its performance is worse than that of the RND policy when rejection has no penalty, corroborating \Cref{thm:flip1-alt}. Under the complete compatibility graph $G_3$, the SB policy improves upon the FR policy, though not by much. 

The four rows in the middle of \Cref{tb:sim_results} show the fraction of jobs lost due to excessive agent rejection out of the total number of jobs lost. The FRfb policy has a range of $3\%$--$15\%$ when $p=0.8$, which is significantly lower than that of the RND policy. This is reassuring; the FRfb policy only sends the jobs to potentially incompatible agents as a last resort, leading to relatively low rejection probabilities. 

The last four rows report the fraction of agents who misreport their types in equilibrium. In line with \Cref{prop:two_type_more_switch}, more agents misreport their types under FRfb compared to FR. On the other hand, SB seems to slightly better incentivize agents to report their true types under the compatibility graph $G_3$, which might explain the throughput improvement over FR.

\begin{table}[htbp]
\footnotesize
\centering
\begin{tabular}{clcccccccccccc}
\toprule
\multicolumn{2}{c}{\textbf{Agent Distribution}}            & \multicolumn{6}{c}{$\alpha_i = 1,\:\forall i\in\agentset$}                                          & \multicolumn{6}{c}{$\alpha_i=|\loc(i)|,\:\forall i\in\agentset$}                                            \\ \cmidrule(r){3-8} \cmidrule(r){9-14}  
\multicolumn{2}{c}{\textbf{Compatibility Graphs}}          & \multicolumn{2}{c}{$G_1$} & \multicolumn{2}{c}{$G_2$} & \multicolumn{2}{c}{$G_3$} & \multicolumn{2}{c}{$G_1$} & \multicolumn{2}{c}{$G_2$} & \multicolumn{2}{c}{$G_3$} \\ \cmidrule(r){3-4} \cmidrule(r){5-6}  \cmidrule(r){7-8} \cmidrule(r){9-10} \cmidrule(r){11-12} \cmidrule(r){13-14}  
\multicolumn{2}{c}{\textbf{Survival Probability ($p$)}}    &   0.8    &    1.0   &   0.8     &   1.0    &    0.8    &   1.0    &    0.8    &   1.0    &    0.8    &   1.0    &    0.8    &    1.0   \\ \cmidrule(r){1-14} 
 \multicolumn{1}{c}{\multirow{4}[0]{*}{\makecell{\textbf{Throughput} \\[0.5mm] (fraction of fluid UB)}}}             & FR   & 0.896 & 0.896 & 0.872 &     0.872        &    0.819     &    0.819    & 0.863 & 0.863 & 0.860 &    0.860  & 0.816          &   0.816        \\
                                           & RND  & 0.760 & 0.885  & 0.809  & 0.880           & 0.750       & 0.851      & 0.795 & 0.869  &  0.830 & 0.878           & 0.789           &  0.865          \\
                                           & FRfb & 0.900 &  0.903 & 0.887  & 0.893         &  0.855          & 0.866        & 0.882 & 0.886  & 0.884 & 0.890        & 0.867           & 0.878          \\
                                           & SB   & 0.896 & 0.896 & 0.872 & 0.872            & 0.825          & 0.825       & 0.863 & 0.863 &  0.860  & 0.860       & 0.822           &  0.822         \\ \cmidrule(r){1-14} 
\multicolumn{1}{c}{\multirow{4}[0]{*}{\makecell{\textbf{Jobs Lost} \\ \textbf{due to Rejection} \\[0.5mm] (fraction of \\ total lost jobs)}}} & FR   &    0.000    &  0.000     &   0.000    &  0.000     &   0.000   &  0.000     &  0.000    &   0.000    &  0.000   &  0.000   &  0.000   & 0.000    \\
                                           & RND  & 0.501 & 0.000 & 0.438 & 0.000 & 0.637          & 0.000      & 0.585 & 0.000 & 0.461 & 0.000            & 0.608            & 0.000          \\
                                           & FRfb & 0.036 &  0.000 & 0.065 & 0.000         & 0.123           & 0.000       & 0.081 & 0.000  & 0.095 & 0.000   & 0.143        & 0.000        \\
                                           & SB   &  0.000  &  0.000         & 0.000      & 0.000        & 0.000        & 0.000           & 0.000        & 0.000         & 0.000        & 0.000       & 0.000         & 0.000          \\ \cmidrule(r){1-14}
\multicolumn{1}{c}{\multirow{4}[0]{*}{\makecell{\textbf{Agents Deviated} \\[0.5mm] (fraction of \\ total agents)}}}          & FR   & 0.079 & 0.079 & 0.232 &  0.232   &   0.677      & 0.677       & 0.271 & 0.271 & 0.371 & 0.371 & 0.803  & 0.803 \\
                                           & RND  &  0.000  & 0.000            & 0.000       &   0.000          &  0.000       &  0.000           & 0.000        & 0.000           & 0.000        & 0.000            & 0.000          & 0.000            \\
                                           & FRfb & 0.111 & 0.143 & 0.267  & 0.294        &  0.723       & 0.729       & 0.332 & 0.353 & 0.418 & 0.443          & 0.858          & 0.868          \\
                                           & SB    & 0.079 & 0.079 & 0.232  & 0.232     & 0.670          & 0.670       & 0.271 & 0.271 & 0.371 & 0.371     & 0.797           & 0.797           \\ \bottomrule
\end{tabular}
\caption{Simulation results.}
\label{tb:sim_results}
\end{table}

 \Cref{fig:virtual} reports the average matching probabilities (agent utilities) aggregated by priority group, where the priority group of a type $i$ agent is defined as $|\loc(i)|$---the number of job types it can serve. Agents in smaller priority groups are prioritized for matching under FR and FRfb. We consider a rejection penalty of $p=0.8$ and compare the FR, FRfb, and RND policies. The matching probability distributions for FR and RND illustrate two extremes: under RND, more flexible agents enjoy higher matching probabilities since they serve more job types, whereas under FR, less flexible agents benefit from prioritization. Interestingly, FRfb strikes a balance between these extremes and improves the matching probabilities for nearly all priority groups relative to both FR and RND: it outperforms RND by allowing agents to %
 be prioritized by the level of flexibility, and it outperforms FR through its fallback option.

\begin{figure}[htbp]
\centering
\begin{subfigure}[b]{0.48\textwidth}
  \centering
\begin{tikzpicture}
  \centering
  \begin{axis}[
        ybar, axis on top,
        height=6.0cm, width=6.0cm,
        bar width=0.12cm,
        tick align=inside,
        enlarge y limits={value=.1,upper},
        ymin=0, ymax=0.8,
        axis x line*=bottom,
        axis y line*=left,
        y axis line style={opacity=0},
        tickwidth=0pt,
        enlarge x limits=true,
        legend image code/.code={
        \draw [#1] (0cm,-0.1cm) rectangle (0.12cm,0.25cm); },
        legend style={
            anchor=north east,
            font=\footnotesize
        },
        xlabel={Priority Group},
        ylabel={Matching Probability},
        label style={font=\footnotesize},
       xtick={1,2,3,4},
    ]
        \addplot table [
    x index=0, 
    y index=1,
    col sep=space
    ] {wait_G3_newutil.txt};
    \addplot table [
    x index=0, 
    y index=2,
    col sep=space
    ] {wait_G3_newutil.txt};
    \addplot table [
    x index=0, 
    y index=3,
    col sep=space
    ] {wait_G3_newutil.txt};
    \legend{FR, FRfb ($p=0.8$), RND ($p=0.8$)}
  \end{axis}
  \end{tikzpicture}
\caption{Agent distribution with $\alpha_i=1,\:\forall i\in\agentset$}
  \label{fig:virtual_1}
\end{subfigure}
~
\begin{subfigure}[b]{0.48\textwidth}
  \centering
\begin{tikzpicture}
  \centering
  \begin{axis}[
        ybar, axis on top,
        height=6.0cm, width=6.0cm,
        bar width=0.12cm,
        tick align=inside,
        enlarge y limits={value=.1,upper},
        ymin=0, ymax=0.8,
        axis x line*=bottom,
        axis y line*=left,
        y axis line style={opacity=0},
        tickwidth=0pt,
        enlarge x limits=true,
        legend image code/.code={
        \draw [#1] (0cm,-0.1cm) rectangle (0.12cm,0.25cm); },
        legend style={
            anchor=north east,
            font=\footnotesize
        },
        xlabel={Priority Group},
        label style={font=\footnotesize},
       xtick=data,
    ]
          \addplot table [
    x index=0, 
    y index=1,
    col sep=space
    ] {wait_G3_flex_newutil.txt};
    \addplot table [
    x index=0, 
    y index=2,
    col sep=space
    ] {wait_G3_flex_newutil.txt};
    \addplot table [
    x index=0, 
    y index=3,
    col sep=space
    ] {wait_G3_flex_newutil.txt};
    \legend{FR, FRfb ($p=0.8$), RND ($p=0.8$)}
  \end{axis}
  \end{tikzpicture}
\caption{Agent distribution with $\alpha_i=|\mathcal{L}(i)|,\:\forall i\in\agentset$}
  \label{fig:virtual_2}
\end{subfigure}
\caption{Matching probabilities (agent utilities) under the complete compatibility graph $G_3$ across different policies.}
\label{fig:virtual}
\end{figure}

\section{Concluding Remarks}
\label{sec:conclusion}

Returning to the affordable-housing allocation example discussed in Section~\ref{sec:introduction}, the system can be modeled with two job types (resale units and new units) and three agent types (flexible, resale-only, and new-only applicants). Under the current mechanism, new units are allocated by lottery among flexible and new-only applicants, while resale units are allocated by lottery among flexible and resale-only applicants, which effectively mimics a compatibility-agnostic random policy (RND)---because applicants can decline offers, it is in equilibrium a weakly dominant strategy to report as flexible, and NYC Housing Connect even encourages applicants to ``check the resale box.''\footnote{\url{https://housingconnect.nyc.gov/PublicWeb/about-us}, accessed 08/17/2025.} The drawback is that applicants who truly prefer only new units can receive multiple resale offers, leading to inefficiencies. The FRfb policy avoids this: it offers resale units first to resale-only applicants, then to flexible applicants, and only as a last resort to self-reported new-only applicants who may be compatible.

On a separate note, one may wonder whether there exists a matching policy that induces completely truthful reporting, as the FRfb policy may still incentivize some agents to under-report their types. Using the example in \Cref{fig:robust_capacity_reservation} with two job types, one might consider a policy that \emph{partially} prioritizes flexible agents—for instance, sending a type~1 job to queue~1 first with a certain probability, fine-tuned so that no type~0 flexible agent has an incentive to misreport. However, determining this probability would require precise knowledge of arrival rates and specific assumptions about agents' strategic behavior, becoming computationally intensive with multiple agent types. In contrast, the FRfb policy is entirely parameter-free and does \emph{not} depend on any equilibrium concept or assumptions about agents' strategic behavior, including their utility functions.

From a practical standpoint, the simple, parameter-free design and robust performance guarantees of the FRfb policy make it a strong candidate for implementation. Consider our ridesharing example: \Cref{fig:fomo} illustrates how FRfb is applied in Uber’s driver destination feature. In \Cref{fig:fomo_1}, a driver in destination mode (reporting as specialized) receives a trip toward her chosen destination, corresponding to dispatching a job of type \(j \neq 0\) to queue \(j\) in the model (see \Cref{fig:aggressive-capacity-reservation}). In \Cref{fig:fomo_2}, the same driver is offered a trip away from her destination, representing a job of type~0 dispatched to a nonzero queue. In both cases, the driver can accept (``tap to accept'') or reject (``no thanks'') based on her true preferences. For the scenario in \Cref{fig:fomo_2}, a genuinely specialized driver will reject the trip, whereas a flexible driver pretending to be specialized will accept it. Under FRfb, such away-from-destination dispatches occur only when no drivers who reported being flexible are available.

\begin{figure}[htbp]
\centering
\begin{subfigure}[b]{0.4\textwidth}
    \centering
    \includegraphics[width=0.6\linewidth]{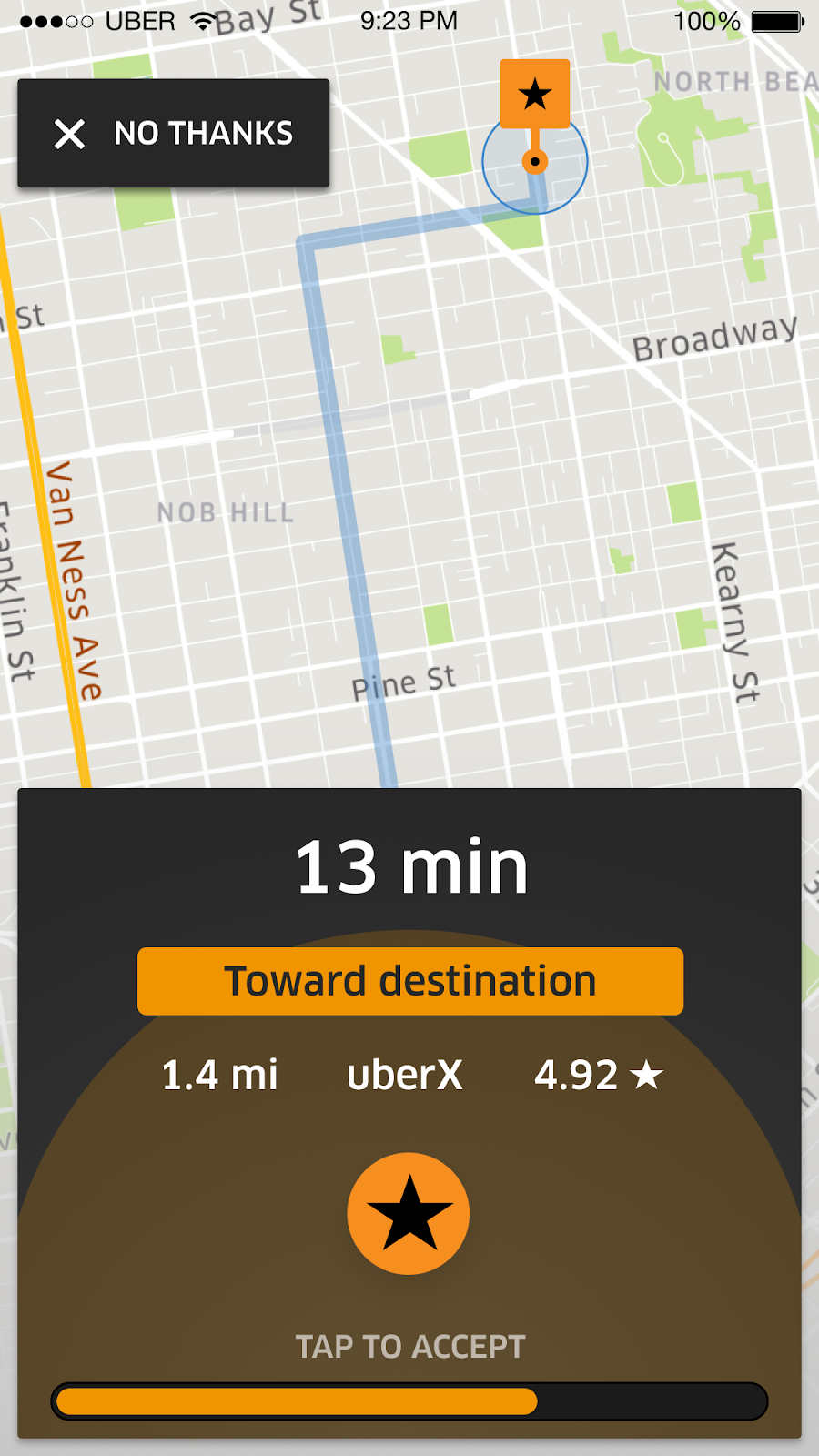}
    \caption{A driver with destination mode receives a ride \emph{toward} her destination}
    \label{fig:fomo_1}
\end{subfigure}
~~~~
\begin{subfigure}[b]{0.4\textwidth}
  \centering
  \includegraphics[width=0.6\linewidth]{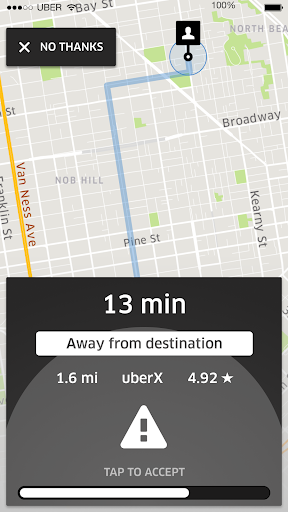}
  \caption{A driver with destination mode receives a ride \emph{away} from her destination}
  \label{fig:fomo_2}
\end{subfigure}
\caption{An implementation of the flexibility reservation with fallback (FRfb) policy in the driver destination product on Uber's driver app.}
\label{fig:fomo}
\end{figure}

\bibliographystyle{informs2014} %
\bibliography{refs}

\medskip

\begin{APPENDICES}

\section{Results for General Compatibility Graph and Auxiliary Formal Results} \label{sec:general_compatibility_graph}

In this appendix section, we present generalized results for compatibility graphs where each agent type can perform any subset of job types, along with several auxiliary formal results.
Let $\loc = \{0, \cdots, \numLoc\}$ be the set of job types and $\agentset = \{0, \cdots, \numAgent\}$ be the set of agent types.
For each job $j\in\loc$, $\agentset(j)$ is the set of agent types that are compatible with job type $j$ and similarly for each agent type $i\in\agentset$, $\loc(i)$ is the set of job types agent type $i$ can serve. For two types of agents $i\neq i^\prime\in\agentset$, we say that agent type $i$ is more flexible than agent type $i^\prime$ (or agent type $i^\prime$ is more specialized than agent type $i$) if $\loc(i)\supset\loc(i^\prime)$.

\subsection{Matching Policies under Full Information} 
\label{appx:fb}
We now formally define the set of matching policies we consider under the full information setting. The policy acts each time a job arrives. Let $j\in\loc$ be an arbitrary type of arriving job. Let $\vecAgent$ indicate the state at the time when this job arrives where %
$\vecAgent \in\mathbb{Z}_{\ge0}^{\numAgent+1}$
contains the number of waiting agents of each type. A matching policy is a function $\typepolicy: \mathbb{Z}_{\ge0}^{\numAgent+1}\times\loc \rightarrow [0,1]^{\numAgent + 1}$, where $\typepolicy_i(\vecQ, j)$ denotes the probability of assigning a type $j$ job to a type $i$ agent, when the number of agents of different types is $\vecQ$.  The set of admissible policy $\typepolicyset$ is defined as: 
\begin{align}
    \label{eq:admissible_policy}
    \typepolicyset \triangleq \left\{\typepolicy : ~~\begin{array}{ll}
    \sum_{i\in\agentset}\typepolicy_i(\vecAgent, j)\le1, & \quad\forall j\in\loc,~\forall\vecAgent\in\mathbb{Z}^{\numAgent+1}_{\ge0} \\ 
     \typepolicy_i(\vecAgent,j)=0, & \quad\forall j\in\loc,~\forall i\in\agentset,~\forall\vecAgent: \vecAgent_i=0 %
     \end{array}\right\}.
\end{align}

The first row of \cref{eq:admissible_policy} requires the admissible policy to define a probability measure over all actions. These actions consist of matching to one of the agent type $i\in\agentset$, together with the option of rejecting a job with probability $1 - \sum_{i \in \agentset} \typepolicy_i(\vecAgent, j)$---this formulation also admits idling policies, where the platform may reject jobs even when compatible agents are available. The second row prohibits assigning jobs to agent types with no available agents.

The following result provides a partial characterization of the optimal matching policy when agents' types are known to the platform (first-best), extending \Cref{prop:fb} to general compatibility graphs and capturing the principle of reserving flexibility.
\begin{proposition}
\label{prop:fb_general}
Among all policies in $\typepolicyset$, there exists an optimal non-idling policy that achieves maximum throughput and satisfies the following flexibility reservation (FR) property: for any job of type $j$, the policy dispatches it to an agent of type $i \in \agentset(j)$ before assigning it to an agent of type $i' \in \agentset(j)$ whenever $\loc(i) \subset \loc(i')$; equivalently, $\typepolicy_{i'}(\vecAgent, j) = 0$ if $\vecAgent_{i} > 0$.
\end{proposition}

\begin{remark}
\label{rmk:FR_nested}
\Cref{prop:fb_general} provides a complete characterization of the first-best matching policy under \emph{nested} compatibility graphs. A compatibility graph is nested if, for each job type $j \in \loc$, the compatible agents $\agentset(j) = \{i, i', i'', \ldots\}$ can be ordered so that $\loc(i) \subset \loc(i') \subset \loc(i'') \subset \cdots$. Under a first-best policy, jobs of type $j$ are dispatched according to this order. Both $G_1$ and $G_2$ in the numerical experiments of \Cref{sec:sim_experiments} are nested.
\end{remark}

\smallskip

To upper bound the maximum throughput attainable by any matching policy under any compatibility graph, we formulate the following linear program, adapted from \cite{aouad2022dynamic}. 
\begin{align}
\tag{\textsf{Fluid}} \label{eq:fluid}
\max\qquad &\sum_{i\in\agentset, j\in\loc(i)} x_{i,j}   \\ \notag
\textrm{s.t.}\qquad &\sum_{i\in\agentset(j)} x_{i,j} \le \mu_j,\quad \forall j\in\loc, \\ \label{eq:supply_constr}
&\sum_{j\in\loc(i)} x_{i,j} + y_i =\lambda_i,\quad \forall i\in\agentset, \\ \notag 
&\theta x_{i,j} \le \mu_j y_i, \quad \forall i\in\agentset,\:\forall j\in\loc(i), \\ \notag
& x_{i,j},\: y_i\ge0, \quad \forall i\in\agentset,\: \forall j\in\loc(i).
\end{align}

\begin{proposition}
\label{prop:fluid_UB}
The optimal value of \eqref{eq:fluid} is greater than or equal to the throughput of any policy $\typepolicy\in\typepolicyset$, for any compatibility graph and model parameters.
\end{proposition}

\smallskip

In addition to providing an upper bound, solving \eqref{eq:fluid} yields a natural heuristic for approximating the first-best policy under general compatibility graphs (used as the FR policy for graph $G_3$ in \Cref{sec:sim_experiments}). Let $(\gamma^\ast_i)_{i \in \agentset}$ denote the optimal dual values corresponding to constraints \eqref{eq:supply_constr}, where $\alpha^\ast_i$ represents the shadow price of an agent of type $i$. The induced matching policy follows a priority-list structure: each job of type $j$ is assigned to an agent of type $i \in \agentset(j)$ in increasing order of $(\gamma^\ast_i)_{i \in \agentset(j)}$, depending on availability (thus reserving the most high-value agents for last). If no agents in $\agentset(j)$ are available, the job is rejected.

\subsection{Flexibility Reservation with Fallback}
\label{appx:frfb}

We now extend our definition of the FRfb policy in \Cref{dfn:FR_fb} to general compatibility graphs, which are used in \Cref{sec:sim_experiments} for $G_2$ and $G_3$.

\smallskip

\begin{definition}[Flexibility Reservation with Fallback]\label{dfn:FR_fb_general}
The Flexibility Reservation with Fallback (FRfb) policy organizes a set of queues, one for each type of agent, $\mathcal{Q}^\RCR = \agentset$. It dispatches a job of type $j\in\loc$ following an order of $\rho^\RCR_j = \left(\mathcal{Q}_{j,1}, \mathcal{Q}_{j,2},\cdots,\mathcal{Q}_{j,K}\right)$ where $\mathcal{Q}_{j,1}, \mathcal{Q}_{j,2},\cdots,\mathcal{Q}_{j,K}$ form a partition of $\mathcal{Q}^\RCR$. In particular, $\mathcal{Q}_{j,1}, \mathcal{Q}_{j,2},\cdots,\mathcal{Q}_{j,K-1}$ form a partition of the set of compatible queues $\agentset(j)$ where $\mathcal{Q}_{j,k} = \{i\in \agentset(j): |\loc(i)|=k\},\:\forall k\in\{1,\cdots,K-1\}$, i.e., the set of compatible queues that can serve exactly $k$ types of jobs. On the other hand, the last subset $\mathcal{Q}_{j,K}$ is constructed as $\mathcal{Q}_{j,K} = \mathcal{Q}^\RCR\setminus\agentset(j)$, i.e., the subset of queues that are incompatible with serving type $j$ jobs. %
\end{definition}

\smallskip

We now give an example of \Cref{dfn:FR_fb_general} using a complete compatibility graph with three job types.

\smallskip

\begin{example}
Consider a complete compatibility graph with 3 %
types of jobs with $\loc=\{0,1,2\}$ and the set of agents (queues) $\agentset$ serving jobs $\{0\}, \{1\}, \{2\}, \{0,1\}, \{1,2\}, \{0,2\}, \{0,1,2\}$. %
We call them agent types (and queue numbers) $0$ through $6$ respectively. %
In this case,
\be
\orderedqueue^\RCR_0 =& \big(\{0\},\: \{3,5\},\: \{6\},\: \{1,2,4\}\big),\\
\orderedqueue^\RCR_1 =& \big(\{1\},\: \{3,4\},\: \{6\},\: \{0,2,5\}\big),\\
\orderedqueue^\RCR_2 =& \big(\{2\},\: \{4,5\},\: \{6\},\: \{0,1,3\}\big).
\ee
\end{example}

\end{APPENDICES}

\ECSwitch

\ECHead{\centering{\textsc{Online Appendix to Matching Queues, Flexibility and Incentives}}}

\section{Proofs and Auxiliary Results} %

\medskip

Given a job offered to a pool of $b$ agents, among whom $a$ are compatible with the job, let $\beta(a,b)\in[0,1]$ denote the probability that the job is successfully matched during this offer sequence. The following lemma summarizes several properties of $\beta(a,b)$ that will be useful in subsequent proofs.

\begin{lemma}
\label{lemma:beta}
Suppose a job is offered to a pool of $b\geq0$ agents in a uniformly random order, without replacement, among whom $a$ with $0\leq a\leq b$ are compatible, and after each rejection, the job survives (is available for another offer) with probability $p\in[0,1]$. Then, in equilibrium, $\beta(a,b)=\frac{a}{b} + \frac{b-a}{b}\cdot p\cdot \beta(a,b-1)$ for $b\geq 1$, and
\begin{itemize}
\item [(i)] $\beta(a,b)$ is non-increasing in $b$ and $\beta(a,a+c)$ is non-decreasing in $a$ for any fixed $c\geq 0$; 
\item [(ii)] $\beta(a,a+c)$ is convex in $c$; %
\item [(iii)] $\beta(0,b) = 0,\:\forall b\ge0$ and $\beta(b,b)=1,\:\forall b>0$.
\end{itemize}
\end{lemma}

The first bullet point states that increasing the total number of agents while keeping the number of compatible agents fixed decreases the probability of success, whereas increasing the number of compatible agents while holding the number of incompatible agents constant increases it. These conditions together imply that $\beta(a, b)$ is non-decreasing in $a$. The second bullet point asserts that the marginal negative impact of additional incompatible agents on the probability of success diminishes as their number increases. Finally, the last bullet point specifies that if all agents are compatible, the probability of success is one, and if none are compatible, the probability is zero.

\smallskip

\begin{proof}{\textbf{Proof of \Cref{lemma:beta}.}} 
We begin by showing the recursion that $\beta(a,b)$ satisfies. First, observe that in equilibrium agents always accept compatible jobs. Because a job is offered to the pool uniformly at random, it is received by a compatible agent (and leaves the system) with probability $a/b$. However, when received by an incompatible agent (which happens with probability $(b-a)/a$) the job survives with probability $p$, at which point it will be matched with probability $\beta(a,b-1)$ (offers are without replacement). 

Note that when there are no compatible agents, the probability of the job being successfully matched must be zero. Moreover, using the recursion it is easy to see that $\beta(b,b)=1$ when $b\geq 1$. This verifies {\it (iii)}.

Next, we verify {\it (i)}. First, note that for $a\geq 1$ we have that $\beta(a,a)=1$, hence $\beta(a,a+1)\leq \beta(a,a)$. Suppose that $\beta(a,b)\leq \beta(a,b-1)$ with $b-1\geq a$. Then, using the recursion, we have
$$
\beta(a,b+1) = \frac{a}{b+1}+\left(1-\frac{a}{b+1}\right) p \beta(a,b) \leq \frac{a}{b}+\left(1-\frac{a}{b}\right) p \beta(a,b-1)=\beta(a,b),
$$
where we have used the induction hypothesis and that $f(x) = x +(1-x)\cdot m$ is non-decreasing when $m\leq 1$. Next, note that for $c=0$, $\beta(a,a+c)=1$ for any $a\geq 1$ (and it equals zero for $a=0$), hence for $c=0$ we have that $\beta(a,a+c)$ is non-decreasing in $a$. Assume this property holds for $c-1$, then 

$$
\beta(a,a+c) = \frac{a}{a+c}+\left(1-\frac{a}{a+c}\right) p \beta(a,a + c-1) \leq \frac{a+1}{a+1+c}+\left(1-\frac{a+1}{a+1+c}\right) p \beta(a+1,a + c)=\beta(a+1,a+1+c),
$$
where we have used the induction hypothesis and that $f(x) = x +(1-x)\cdot m$ is non-decreasing when $m\leq 1$.

Finally, we verify {\it (ii)}. Note that for $a=0$, $\beta(a,a+c)$ is constant and thus convex. Let's consider $a\geq 1$ and $c\geq 0$. To ease notation, let $g(c) = \beta(a,a+c)$, $q_c = a/(a+c)$, and $r_c = pc/(a+c)$. Let's also define $\Delta_c = g(c)-g(c-1)$, then we need to show that $\Delta_{c+1} -\Delta_c\geq 0$. We have
\begin{flalign*}
\Delta_{c+1} -\Delta_c &= (q_{c+1} +r_{c+1}g(c) - g(c)) -(g(c)-g(c-1))\\
&= q_{c+1} +(r_{c+1}-2)g(c) + g(c-1)\\
&= q_{c+1} +(r_{c+1}-2)(q_c +r_c g(c-1)) + g(c-1)\\
&= q_{c+1} +(r_{c+1}-2)q_c  + ((r_{c+1}-2)r_c+1)g(c-1)\\
&= \underbrace{((r_{c+1}-2)r_c+1)}_{(A)}\underbrace{\left(g(c-1)
+\frac{q_{c+1} +(r_{c+1}-2)q_c}{((r_{c+1}-2)r_c+1)}
\right)}_{(B)}.
\end{flalign*}
We need to verify that $(A)$ and $(B)$ are non-negative. We have 
$$
(A) = \frac{a + a^2 + 2 a c (1 - p) + c (1 + c) (1 - p)^2}{(a+c)(a+c+1)}\geq 0.
$$
For $(B)$, we proceed by induction. Let $s_c$ be such that $(B)=g(c-1)+s_c$. For $c=1$, we have $g(0)+s_1 = 2 (1 - p)^2/(a^2 + a (3 - 2 p) + 2 (1 - p)^2)\geq 0$. Assuming that $g(c-2)+s_{c-1}\geq 0$, then $g(c-1)+s_c \geq q_{c-1}-r_{c-1}s_{c-1} +s_c$ and this is equal to
$$
\frac{2 a (a + c) (1 - p)^2}{(a + a^2 + 2 a c (1 - p) + c (1 + c) (1 - p)^2) (a(a-1)+(c-1)c(1-p)^2+2a(c(1-p)+p))}\geq 0,
$$
which concludes the proof. $\square$
\end{proof}

\medskip

\begin{proof}{\textbf{Proof of \Cref{lm:existence_eq}.}}
The proof follows similar steps as Theorem 1 in \citet{smith1979existence} for proving the existence of Wardrop equilibrium. Fix a policy $\policy$, we first give an equivalent definition of $\sigma$ being an equilibrium, aside from \cref{eq:equilibrium}.
\begin{equation*}
    \sum_{i\in \agentset,q\in \queueset}u_{i,q}^\policy(\joinprob)%
    (\joinprob_{i,q}-\tilde{\joinprob}_{i,q})\geq 0,\quad \forall \tilde{\sigma}\in \joinprobset(\queueset) ~ \Longleftrightarrow ~ \textrm{for all $i\in\agentset$ and $q\in\queueset$}, \:\:\joinprob_{i,q}>0 \: \Rightarrow\:  u_{i,q}^{\policy}(\joinprob)\:\ge\: u_{i,q'}^{\policy}(\joinprob),\:\: \forall q'\in\queueset.
\end{equation*}
We prove this equivalence:
\begin{itemize}
    \item ``$\Longleftarrow$''
    \begin{align*}
        &\sum_{i\in\agentset, q\in\queueset}u_{i,q}^{\policy}(\joinprob)\joinprob_{i,q} \\
        =&\sum_{i\in\agentset, q\in\queueset:\joinprob_{i,q}>0}u_{i,q}^\policy(\joinprob)\joinprob_{i,q} \\ 
        =&\sum_{i\in\agentset, q\in\queueset:\joinprob_{i,q}>0} \left(\max_{q'\in\queueset} u_{i,q'}^\policy(\joinprob)\right) \joinprob_{i,q} \\
        =&\sum_{i\in\agentset, q\in\queueset}\left(\max_{q'\in\queueset}u_{i,q'}^\policy(\joinprob)\right)\tilde{\joinprob}_{i,q} \\
        \ge&\sum_{i\in\agentset, q\in\queueset}u_{i,q}^\policy(\joinprob)\tilde{\joinprob}_{i,q}.
    \end{align*}
    
    \item ``$\Longrightarrow$'' 
    
    For this direction, we prove it by contradiction. Suppose that there exists a $\joinprob_{i^*,q^*}>0$ and $u_{i^*,q^*}^{\policy}(\joinprob)<u_{i^*,q}^\policy(\joinprob), \forall q\in\queueset$. Now consider a $\tilde{\joinprob}$ which is constructed as follows:
    \begin{equation*}
    \tilde{\joinprob}_{i,q} = \left.
      \begin{cases}
        \joinprob_{i,q}, & \forall i\neq i^*,\:\forall q, \\
        0, & i=i^*,\:q=q^*, \\
        \joinprob_{i,q}+\frac{\joinprob_{i^*,q^*}}{|\queueset|}, & i=i^*,\:\forall q\neq q^*.
      \end{cases}\right.
    \end{equation*}
    It can be seen that $\sum_{i\in\agentset,q\in \queueset}u_{i,q}^\policy(\joinprob)(\joinprob_{i,q}-\tilde{\joinprob}_{i,q})<0$, which reaches a contradiction.
\end{itemize}

This equivalent condition says that the vector $u^\policy(\joinprob)$ with components $u^{\policy}_{i,q}(\joinprob)$ is normal to the simplex
$\joinprobset(\queueset)$ at $\sigma$. Let $P_{\Sigma}(\cdot)$ be the projection operator onto $\joinprobset(\queueset)$. 
Define the mapping $f:\joinprobset(\queueset)\rightarrow \joinprobset(\queueset)$ by $f(\sigma)=P_{\Sigma}(\sigma-u^{\policy}_{i,q}(\joinprob))$, then the equilibrium condition can be cast as the fixed point equation
$f(\sigma)=\sigma$. Since, by assumption,
$u^\policy(\sigma)$ is continuous, and the projection operator onto a bounded convex set is continuous, then $f$ is a continuous function as the composition of continuous functions is continuous. By Brouwer's fixed-point theorem, we deduce that an equilibrium always exists. This completes the proof.~$\square$
\end{proof}

\medskip

Let $\vece_{i} \in \setZ^{|\agentset|}_{\geq 0}$ be the unit vector which is zero except for the $i\th$ element. Recall $\typepolicyset$ is the space of state-dependent policies under full information defined in \eqref{eq:admissible_policy}. For any policy $\typepolicy\in\typepolicyset$, $M^\typepolicy(T)$ is the total number of matches made up to time $T>0$ under policy $\typepolicy$ and let $\mathbb E[M^\typepolicy(T)\mid\vecAgent\supZero]$ be the expected number of matches up to time $T$ given that the initial number of agents on the platform is $A\in\mathbb Z_{\ge0}^{|\agentset|}$. The following lemma is critical in proving \Cref{prop:fb_general} (and \Cref{prop:fb}).

\begin{lemma}[Value of a More Flexible Agent] \label{lm:v_flexibility}
The following holds
\begin{equation}\label{eq:key-property}
    1 + \sup_{\phi\in\Phi}\:\E[M^{\typepolicy}(T)\mid \vecAgent\supZero = \vecAgent] \ge \sup_{\phi\in\Phi}\:\E[M^{\typepolicy}(T)\mid \vecAgent\supZero = \vecAgent + \vece_{i}] \geq \sup_{\phi\in\Phi}\:\E[M^{\typepolicy}(T)\mid \vecAgent\supZero = \vecAgent + \vece_{i^\prime}],
\end{equation}
for any time horizon $T\geq 0$, any $\agent \in \setZ^{|\agentset|}_{\geq 0}$ and any agent types $i\neq i^\prime\in\agentset$ such that $\loc(i)\supset\loc(i^\prime)$.
\end{lemma}

\begin{proof}{\textbf{Proof of \Cref{lm:v_flexibility}.}} 
Let $M^\phi_k(A,t)$ denote the expected number of matches under policy $\phi$ starting with state $A$ from time $t$ up until the time horizon $T$ expires or until we reach a total of $k$ events, whichever occurs first. Events include job arrivals, agent arrivals, and reneging. 

Denote by $M^\phi_\infty(A,0) = \E[M^\phi(T)\mid \vecAgent\supZero = \vecAgent]$ the expected number of matches under policy $\phi$ within $[0,T]$ without a constraint on the total number of events. We will show that $1 + \sup_{\phi\in\Phi} M^\phi_\infty(A,t) \ge \sup_{\phi\in\Phi} M^\phi_\infty(A+e_i,t) \ge \sup_{\phi\in\Phi} M^\phi_\infty(A+e_{i^\prime},t)$ for all $i\neq i^\prime\in \agentset$ such that $\loc(i)\supset\loc(i^\prime)$ and all $t\in[0,T]$.

For any finite $k$ and any state $A$, $|\sup_{\phi\in\Phi} M^\phi_\infty(A,t) - \sup_{\phi\in\Phi} M^\phi_k(A,t)|$ is bounded by the expected number of events (job arrival events in particular) in excess of $k$ over $[t,T]$. This upper bound vanishes as $k$ grows large.
Thus, $\lim_{k\to\infty} \sup_{\phi\in\Phi} M^\phi_k(A,t) = \sup_{\phi\in\Phi} M^\phi_\infty(A,t)$ for each $A$ and $t$. Thus, to show our claim, it is sufficient to show that $1 + \sup_{\phi\in\Phi} M^\phi_k(A,t) \ge \sup_{\phi\in\Phi}M^\phi_k(A+e_i,t) \ge \sup_{\phi\in\Phi} M^\phi_k(A+e_{i^\prime},t)$ for all finite $k$. We now show this via induction on $k$.

This induction hypothesis is trivially true for $k=0$ and so we turn our attention to show it is true for $k>0$ assuming that it is true for $k-1$.
Consider three systems at time $t$, one with state $A$, another with state $A+e_i$ and the third with state $A+e_{i^\prime}$.
We will refer to the three systems as $S_0$, $S_1$, and $S_2$ respectively.

We fix a sequence of agent and job arrivals within $[t,T]$ (i.e., we condition on the sequence of types and the arrival times of agents and jobs). Also consider a collection of independent exponential random variables representing the times until each of several possible reneging events occur in the system. Then the next event occurs randomly within the following categories:

\smallskip

\begin{enumerate}
\item reneging by an agent counted in $A$, from a type with a non-zero component in $A$;
\item an agent arrival, from one of the types in $\agentset$ (the arrival time of next agent is not random but fixed);
\item a job arrival, from one of the types in $\loc$ (the arrival time of next job is not random but fixed);
\item reneging by the agent not counted in $A$, which is of type $i$ in system $S_1$ and type $i^\prime$ in system $S_2$.
\end{enumerate}

\smallskip

We let random variable $t'$ represent the sum of $t$ and the minimum of these times. The event achieving this minimum determines the event that occurs next. If $T$ occurs before any of these times, then the time horizon expires before the next event occurs. Regarding reneging of the agents not counted in $A$, note that reneging occurs at the same rate regardless of agent type, so it is valid to couple them with a single event.

For each system $S_n, n=0,1,2$, let random variable $m(n)$ represent whether a match results from the next event, let $A'(n)$ represent the system state after the event occurs, and let $\Delta$ indicate whether the next event occurred in all three systems ($\Delta=1$) or just systems 1 and 2 ($\Delta=0$, which occurs on event type 4).
Thus,
\begin{equation*}
\begin{split}
\sup_{\phi\in\Phi}M^\phi_k(A,t) = \E\left[m(0) + \sup_{\phi\in\Phi} M^\phi_{k-\Delta}(A'(0),t')\right]
\ge \E\left[m(0) + \sup_{\phi\in\Phi} M^\phi_{k-1}(A'(0),t')\right],\\
\sup_{\phi\in\Phi}M^\phi_k(A+e_i,t) = \E\left[m(1) + \sup_{\phi\in\Phi} M^\phi_{k-1}(A'(1),t')\right],\\
\sup_{\phi\in\Phi} M^\phi_k(A+e_{i^\prime},t) = \E\left[m(2) + \sup_{\phi\in\Phi} M^\phi_{k-1}(A'(2),t')\right].
\end{split}
\end{equation*}

The expectations are taken over $m(n), A'(n)$ and $t'$. To show our result, we will consider the next event on a case-by-case basis to show that with probability one,
\begin{equation}
\label{eq:alt-proof}
1 + m(0) + \sup_{\phi\in\Phi}M^\phi_{k-1}(A'(0),t')
\ge  m(1) + \sup_{\phi\in\Phi}M^\phi_{k-1}(A'(1),t')
\ge  m(2) + \sup_{\phi\in\Phi}M^\phi_{k-1}(A'(2),t').
\end{equation}
If $t'>T$, then $\sup_{\phi\in\Phi} M^\phi_{k-1}(A'(n),t')=0, m(n) = 0$ for all $n$, automatically verifying this expression. Thus, it is sufficient to focus on $t'\le T$.

\smallskip

\emph{\underline{Case 1}}: The next event is reneging by an agent counted in $A$ or an agent arrival.  In this case, let $A''$ represent $A$ modified by this event, so that $A'(0) = A''$, $A'(1) = A'' + e_i$ and $A'(2) = A''+ e_{i^\prime}$.
Also, $m(n) = 0$ for all $n$.
Then, by the induction hypothesis,
$1 + \sup_{\phi\in\Phi} M^\phi_{k-1}(A'', t')
\ge \sup_{\phi\in\Phi} M^\phi_{k-1}(A'' + e_i, t')
\ge \sup_{\phi\in\Phi} M^\phi_{k-1}(A'' + e_{i^\prime}, t')$,
showing \eqref{eq:alt-proof}.

\medskip

\emph{\underline{Case 2}}: The next event is reneging by an agent not counted in $A$. 
In this case, $A'(0) = A'(1) = A'(2) = A$ and $m(0) = m(1) = m(2) = 0$.
Thus, \eqref{eq:alt-proof} is verified directly, with $1+ \sup_{\phi\in\Phi} M^\phi_{k-1}(A,t')\ge \sup_{\phi\in\Phi}M^\phi_{k-1}(A,t')$.

\medskip

\emph{\underline{Case 3}}: The next event is a job type that is not compatible with both agents of types $i$ and $i^\prime$. In this case, we can couple the matching decision so that all three systems match the job with an agent in $A$ or all three systems choose not to match. In either case, the induction step holds.

\medskip

\emph{\underline{Case 4}}: The next event is a job type that is compatible with agent type $i'$ thus compatible with type $i$ as well. We consider the following subcases. 

\begin{itemize}
\item[(i)] If $S_0$ and $S_1$ (or $S_1$ and $S_2$) make the same matching decision, then the inequality involving the comparison of $S_0$ and $S_1$ (or $S_1$ and $S_2$) in \eqref{eq:alt-proof} holds. In the next two subcases we consider $S_0$ and $S_1$, $S_1$ and $S_2$ make different matching decisions respectively. 

\smallskip

\item[(ii)] If $S_0$ and $S_1$ make different matching decisions, consider the first scenario where $S_0$ decides to match an agent of type $j$ and $S_1$ decides to match an agent of type $i$ and $j\neq i$. The first part of the inequality in \eqref{eq:alt-proof} can be verified,
\begin{align}
1 + m(0) + \sup_{\phi\in\Phi}M^\phi_{k-1}(A'(0),t') =& 1 + 1 + \sup_{\phi\in\Phi} M_{k-1}^\phi(A - e_j, t') \notag \\
\ge& 1 + \sup_{\phi\in\Phi} M_{k-1}^\phi(A, t')  \tag{optimality of $\sup_{\phi\in\Phi} M_k^\phi(A, t)$} \\
=& 1 + \sup_{\phi\in\Phi} M_{k-1}^\phi(A + e_i - e_i, t') = m(1) + \sup_{\phi\in\Phi}M^\phi_{k-1}(A'(1),t'). \notag
\end{align}
Now consider the second scenario where $S_0$ decides to match an agent of type $j$ and $S_1$ decides to match an agent of type $j^\prime\neq j\neq i$. %
\begin{align}
1 + m(0) + \sup_{\phi\in\Phi}M^\phi_{k-1}(A'(0),t') = &1 + 1 + \sup_{\phi\in\Phi} M_{k-1}^\phi(A - e_j, t') \notag \\
\ge& 1 + 1 + \sup_{\phi\in\Phi} M_{k-1}^\phi(A - e_{j'}, t')  \tag{optimality of $\sup_{\phi\in\Phi} M_k^\phi(A, t)$} \\
\ge& 1 + \sup_{\phi\in\Phi} M_{k-1}^\phi(A + e_i - e_{j^\prime}, t') \tag{induction hypothesis at $k-1$}\\
=& m(1) + \sup_{\phi\in\Phi}M^\phi_{k-1}(A'(1),t'). \notag
\end{align}
For the third scenario where $S_0$ decides not to match but $S_1$ decides to match with a type $i$ agent, trivially,
\be
1 + m(0) + \sup_{\phi\in\Phi}M^\phi_{k-1}(A'(0),t') =& 1 + \sup_{\phi\in\Phi} M_{k-1}^\phi(A, t') \\
\ge& 1 + \sup_{\phi\in\Phi} M_{k-1}^\phi(A + e_i - e_i, t') = m(1) + \sup_{\phi\in\Phi}M^\phi_{k-1}(A'(1),t').
\ee
For a fourth scenario where $S_0$ decides not to match but $S_1$ decides to match with a type $j\neq i$ agent, 
\begin{align}
1 + m(0) + \sup_{\phi\in\Phi}M^\phi_{k-1}(A'(0),t') =& 1 + \sup_{\phi\in\Phi} M_{k-1}^\phi(A, t') \notag \\
\ge& 1 + 1 + \sup_{\phi\in\Phi} M_{k-1}^\phi(A - e_j, t') \tag{optimality of $\sup_{\phi\in\Phi} M_k^\phi(A, t)$} \\
\ge& 1 + \sup_{\phi\in\Phi} M_{k-1}^\phi(A + e_i - e_j, t') \tag{induction hypothesis at $k-1$} \\
=& m(1) + \sup_{\phi\in\Phi}M^\phi_{k-1}(A'(1),t'). \notag
\end{align}
For the last scenario where $S_0$ decides to match with an agent of type $j$ but $S_1$ decides not to match,
\begin{align}
1 + m(0) + \sup_{\phi\in\Phi}M^\phi_{k-1}(A'(0),t') =&1 + 1 + \sup_{\phi\in\Phi} M_{k-1}^\phi(A - e_j, t') \notag \\
\ge& 1 + \sup_{\phi\in\Phi} M_{k-1}^\phi(A, t') \tag{optimality of $\sup_{\phi\in\Phi} M_k^\phi(A, t)$} \\
\ge& \sup_{\phi\in\Phi} M_{k-1}^\phi(A+e_i, t') \tag{induction hypothesis at $k-1$} \\
=& m(1) + \sup_{\phi\in\Phi}M^\phi_{k-1}(A'(1),t'). \notag
\end{align}

\smallskip

\item[(iii)] If $S_1$ and $S_2$ make different matching decisions, consider the first scenario where $S_1$ decides to match with an agent of type $j$ while $S_2$ decides to match to an agent of type $j^\prime\neq j$, and assumes that $S_1$ contains a type $j'$ agent. The second part of the inequality in \eqref{eq:alt-proof} can be verified,
\begin{align}
m(1) + \sup_{\phi\in\Phi}M^\phi_{k-1}(A'(1),t') =& 1 + \sup_{\phi\in\Phi}M^\phi_{k-1}(A + e_i - e_j,t')  \notag\\
\ge& 1 + \sup_{\phi\in\Phi}M^\phi_{k-1}(A + e_i - e_{j'},t') \tag{optimality of $\sup_{\phi\in\Phi} M_k^\phi(A, t)$} \\
\ge& 1 + \sup_{\phi\in\Phi}M^\phi_{k-1}(A + e_{i'} - e_{j'},t') \tag{induction hypothesis at $k-1$} \\
=&m(2) + \sup_{\phi\in\Phi}M^\phi_{k-1}(A'(2),t'). \notag
\end{align}

Consider the second scenario where $S_1$ decides to match with an agent of type $j$ while $S_2$ decides to match to an agent of type $j^\prime\neq j$, and assumes that $S_1$ does not contain a type $j'$ agent. Then it must be that $j'=i'$. This suggests that the incoming job is compatible with a type $i$ agent as well because $\loc(i)\supset\loc(i')$.
\begin{align}
m(1) + \sup_{\phi\in\Phi}M^\phi_{k-1}(A'(1),t') =& 1 + \sup_{\phi\in\Phi}M^\phi_{k-1}(A + e_i - e_j,t')  \notag\\
\ge& 1 + \sup_{\phi\in\Phi}M^\phi_{k-1}(A + e_i - e_i,t') \tag{optimality of $\sup_{\phi\in\Phi} M_k^\phi(A, t)$} \\
=& 1 + \sup_{\phi\in\Phi}M^\phi_{k-1}(A + e_{i'} - e_{i'},t') \notag \\
=&m(2) + \sup_{\phi\in\Phi}M^\phi_{k-1}(A'(2),t'). \notag
\end{align}

For the third scenario where $S_1$ decides to match with an agent of type $j$ while $S_2$ decides not to match,
\begin{align}
m(1) + \sup_{\phi\in\Phi}M^\phi_{k-1}(A'(1),t') =& 1 + \sup_{\phi\in\Phi}M^\phi_{k-1}(A + e_i - e_j,t')  \notag\\
\ge& \sup_{\phi\in\Phi}M^\phi_{k-1}(A + e_i,t') \tag{optimality of $\sup_{\phi\in\Phi} M_k^\phi(A, t)$} \\
\ge& \sup_{\phi\in\Phi}M^\phi_{k-1}(A + e_{i'},t')  \tag{induction hypothesis at $k-1$} \\
=&m(2) + \sup_{\phi\in\Phi}M^\phi_{k-1}(A'(2),t'). \notag
\end{align}

For the fourth scenario where $S_1$ decides not to match while $S_2$ decides to match with an agent of type $j'$ and $S_1$ contains a type $j'$ agent,
\begin{align}
m(1) + \sup_{\phi\in\Phi}M^\phi_{k-1}(A'(1),t') =& \sup_{\phi\in\Phi}M^\phi_{k-1}(A + e_i,t')  \notag\\
\ge& 1 + \sup_{\phi\in\Phi}M^\phi_{k-1}(A + e_i - e_{j'},t') \tag{optimality of $\sup_{\phi\in\Phi} M_k^\phi(A, t)$} \\
\ge& 1 + \sup_{\phi\in\Phi}M^\phi_{k-1}(A + e_{i'} - e_{j'},t')  \tag{induction hypothesis at $k-1$} \\
=&m(2) + \sup_{\phi\in\Phi}M^\phi_{k-1}(A'(2),t'). \notag
\end{align}

For the last scenario where $S_1$ decides not to match while $S_2$ decides to match with an agent of type $j'$ and $S_1$ does not contain a type $j'$ agent, this indicates that $j'=i'$. This further suggests that the incoming job is compatible with a type $i$ agent as well
because $\loc(i)\supset\loc(i')$.
\begin{align}
m(1) + \sup_{\phi\in\Phi}M^\phi_{k-1}(A'(1),t') =& \sup_{\phi\in\Phi}M^\phi_{k-1}(A + e_i,t')  \notag\\
\ge& 1 + \sup_{\phi\in\Phi}M^\phi_{k-1}(A + e_i - e_i,t') \tag{optimality of $\sup_{\phi\in\Phi} M_k^\phi(A, t)$} \\
=& 1 + \sup_{\phi\in\Phi}M^\phi_{k-1}(A + e_{i'} - e_{i'},t') \notag \\
=&m(2) + \sup_{\phi\in\Phi}M^\phi_{k-1}(A'(2),t'). \notag
\end{align}

\end{itemize}

\emph{\underline{Case 5}}: The next event is a job type that is compatible with agent type $i$ but not with the type $i'$. This essentially follows the same proof as Case $4$ with the same subcases (i) and (ii). For subcase (iii), we simply remove the discussion of the scenario indicating that $S_2$ matches the incoming job with a type $i'$ agent. This completes the proof.~ $\square$

\end{proof}

\medskip

\begin{proof}{\textbf{Proof of \Cref{prop:fb_general}.}} 
Suppose that there exists an optimal policy $\phi^\prime$ that does not satisfy the FR property. %
Then at some moment $t$, for two types of agents $i\neq i^\prime\in\agentset$ such that $\loc(i)\supset\loc(i^\prime)$ with $\agent_i>0$ and $\agent_{i^\prime}>0$, policy $\phi^\prime$ assigns an incoming job to a compatible type $i$ agent instead of a compatible type $i^\prime$ agent. Now consider an alternative action that assigns the same job to a type $i^\prime$ agent. By the second inequality of \Cref{lm:v_flexibility}, continuing running the optimal policy from this moment onward for any time duration $T$ weakly improves throughput in expectation. This suggests that any optimal policy that does not satisfy the FR property can be modified to satisfy such a property without decreasing its throughput over any time horizon. Similarly, suppose that there exists another optimal policy $\phi''$ that does not match an incoming job even if there are compatible agents available. Suppose that this event occurs at some time $t$, and consider an alternative action that assigns the job to any compatible agent. Similarly, by the first inequality of \Cref{lm:v_flexibility}, continuing running the optimal policy from this moment onward for any time duration $T$ does not decrease the throughput in expectation. This indicates that any idling optimal policy can be made non-idling without deteriorating its performance. This concludes the proof.~$\square$
\end{proof}

\medskip

\begin{proof}{\textbf{Proof of \Cref{prop:fb}.}} 
\Cref{prop:fb} is a special case of \Cref{prop:fb_general} as any matching policy $\pi=(\queueset, \rho)$ defined in \Cref{subsec:matching} can be represented by a policy $\phi\in\Phi$ when agent types are not private. Its proof thus follows from the proof of \Cref{prop:fb_general}. $\square$
\end{proof}

\medskip

\begin{proof}{\textbf{Proof of \Cref{prop:two_type_more_switch}.}}
Recall that $W^{\FLQR}_{0,q}(\joinprob)$ and $W^{\FLQ}_{0,q}(\joinprob)$ denote the steady-state virtual waiting time of a type 0 agent in queue $q$ under the FRfb and FR policies, respectively, when the queue-joining strategy is $\joinprob$. Since we restrict attention to two types, $\loc=\{0,1\}$, the strategy profile $\joinprob$ can be represented by $\joinprob_{0,1}$. For convenience, we slightly abuse notation and write $W_{0,q}^\policy(\sigma_{0,1})$ in place of $W_{0,q}^\policy(\sigma)$.  

Note that $W^{\FLQ}_{0,1}(\joinprob_{0,1})$ is stochastically increasing in $\joinprob_{0,1}$; that is, $W^{\FLQ}_{0,1}(\joinprob_{0,1}) \ge_{\mathrm{st}} W^{\FLQ}_{0,1}(\joinprob'_{0,1})$ for all $\joinprob_{0,1}\ge\joinprob'_{0,1}$, where $\ge_{\mathrm{st}}$ denotes first-order stochastic dominance. This monotonicity arises because, under the ACR policy, increasing $\joinprob_{0,1}$ raises the arrival rate of agents to queue~1 without affecting the job arrival rate to that queue. This further implies that $u^{\FLQ}_{0,1}(\joinprob_{0,1}) = (v+c/\theta)\mathbb E[e^{-\theta W^{\FLQ}_{0,1}(\joinprob_{0,1})}] \le u^{\FLQ}_{0,1}(\joinprob'_{0,1})=(v+c/\theta)\mathbb E[e^{-\theta W^{\FLQ}_{0,1}(\joinprob'_{0,1})}],\:\forall \joinprob_{0,1}\ge \joinprob'_{0,1}$. Similarly, for any $\joinprob_{0,1}\in[0,1]$  we must have $W^{\FLQ}_{0,1}(\joinprob_{0,1})\geq_{\textrm{st}} W^{\FLQR}_{0,1}(\joinprob_{0,1})$. To see why this is true, consider a sample path argument where we tag a type 0 agent arriving to queue 1 in two systems---one runs under the FRfb policy and the other runs under the FR policy---that encounter the same queue length in queue 1. Let us further consider the time it takes the tagged agent to be matched. Under the FR policy, the tagged agent can only be matched to some job of type 1 but not type 0. In contrast, under the FRfb policy, the tagged agent has a chance to be matched to a job of type 0 as long as queue 0 is empty, on top of being possibly matched to type 1 job. In turn, the tagged agent on any sample path will wait less under the FRfb policy than under the FR policy. That is, $W^{\FLQ}_{0,1}(\joinprob_{0,1})\geq_{\textrm{st}} W^{\FLQR}_{0,1}(\joinprob_{0,1})$ for all $\joinprob_{0,1}$, which further implies $u^{\FLQ}_{0,1}(\joinprob_{0,1})\le u^{\FLQR}_{0,1}(\joinprob_{0,1})$ for all $\joinprob_{0,1}$. Finally, we argue that $W^{\FLQ}_{0,0}(\joinprob_{0,1}) =_{\textrm{st}} W^{\FLQR}_{0,0}(\joinprob_{0,1})$ for all $\joinprob_{0,1}\in [0,1]$. Similar to the sample path argument in the previous paragraph, consider a tagged type $0$ agent arriving to queue 0, we show that this arrival experiences the same waiting time in systems under two policies for each sample path. Indeed, simply notice that the allocation of jobs to agents in queue 0, when starting from the same state, is the same under both policies. The only difference occurs when queue 0 is empty, but at that moment, the tagged agent has already left the system. This further implies $u^{\FLQ}_{0,0}(\joinprob_{0,1}) = u^{\FLQR}_{0,0}(\joinprob_{0,1})$ for all $\joinprob_{0,1}$. 

\begin{figure}[htbp]
\centering
\scalebox{0.65}{\input{figures/Lemma4.tikz}}
\vspace{3mm}\caption{Illustration for the Proof of \Cref{prop:two_type_more_switch}.}
\label{fig:lm_4}
\end{figure} 

\Cref{fig:lm_4} illustrates these utilities which reflect the relationship we just proved above. $u_{0,0}^{\FLQR}(\joinprob_{0,1})$ collapses with $u_{0,0}^{\FLQ}(\joinprob_{0,1})$ for all $\joinprob_{0,1}$ (blue curve), while $u_{0,0}^{\FLQ}(\joinprob_{0,1})\ge u_{0,0}^{\FLQR}(\joinprob_{0,1})$ for all $\joinprob_{0,1}$ (black solid and dashed curves). To conclude the proof, we consider three cases below based on the relative value of $u_{0,0}^{\FLQR}(\joinprob_{0,1})$ (and $u_{0,0}^{\FLQ}(\joinprob_{0,1})$), illustrated by the three different blue curves in \Cref{fig:lm_4}.
\smallskip
\begin{itemize}
    
    \item[\emph{\underline{Case 1}}:]  $u^{\FLQ}_{0,0}(\joinprob_{0,1}) = u^{\FLQR}_{0,0}(\joinprob_{0,1}) \le u_{0,1}^{\FLQ}(\joinprob_{0,1}),\:\forall\joinprob_{0,1}$:
    In such a case, %
    we can deduce that the only equilibrium under the FR and the FRfb policies are $\joinprobFLQ_{0,1}=\joinprobFLQR_{0,1}=1$, respectively.

    \smallskip
    
    \item[\emph{\underline{Case 2}}:]  $u_{0,1}^{\FLQ}(1)< u^{\FLQ}_{0,0}(0) = u^{\FLQR}_{0,0}(0) < u_{0,1}^{\FLQ}(0)$, there exists $\joinprob^*_{0,1}\in [0,1]$ such that 
    $u^{\FLQ}_{0,0}(\joinprob_{0,1}^*) = u^{\FLQ}_{0,1}(\joinprob_{0,1}^*)$ and $u^{\FLQ}_{0,0}(\joinprob_{0,1}) < u^{\FLQ}_{0,1}(\joinprob_{0,1})$ for all $\joinprob_{0,1}<\joinprob^\ast_{0,1}$:
     In such a case, $\joinprobFLQ_{0,1}=\joinprob_{0,1}^*$ must be an equilibrium, while clearly $\joinprobFLQR_{0,1}\ge\joinprob_{0,1}^*$ because for all $\joinprob_{0,1}<\joinprob_{0,1}^*$ we have $u^{\FLQR}_{0,0}(\joinprob_{0,1}) = u^{\FLQ}_{0,0}(\joinprob_{0,1}) < u^{\FLQ}_{0,1}(\joinprob_{0,1})<u^{\FLQR}_{0,1}(\joinprob_{0,1})$.

     \smallskip
    
    \item[\emph{\underline{Case 3}}:]  $u^{\FLQ}_{0,0}(0) = u^{\FLQR}_{0,0}(0)\ge u_{0,1}^{\FLQ}(0)$: 
    In such a case, $\joinprobFLQ_{0,1}=0$ is clearly one equilibrium under the FR policy, thus we have $\joinprobFLQR_{0,1}\ge0=\joinprobFLQ_{0,1}$ for some equilibria under the two policies. This concludes the proof. $\square$
\end{itemize}
\end{proof}

\medskip

We next present a similar lemma to \Cref{lm:v_flexibility} that concerns the value of (non)idling under random.
\begin{lemma}
\label{lm:v_flexibility_2}
For any compatibility graph, under the random policy, 
\begin{equation}\label{eq:key-property}
1+ \E[M^{\normalfont{\SMQ}}(T)\mid \agent\supZero = \agent] \geq \E[M^{\normalfont{\SMQ}}(T)\mid \agent\supZero = \agent + \vece_{i}], %
\end{equation}
for all $i\in\agentset$, any time horizon $T\geq 0$ and any $\agent\in\setZ^{|\agentset|}_{\geq 0}$.
\end{lemma}
\begin{proof}{Proof.}

As we fix the policy to be the random policy, we suppress the policy notation in the rest of the proof.
Let $M_k(A,t)$ denote the expected number of matches starting with state $A$ from time $t$ up until the time horizon $T$ expires or until we reach a total of $k$ events, whichever occurs first. Events include job arrivals, agent arrivals, and reneging. As the proof of \Cref{lm:v_flexibility},

For any finite $k$ and any state $A$, $|M_\infty(A,t) - M_k(A,t)|$ is bounded by the expected number of events (job arrival events in particular) in excess of $k$ over $[t,T]$. This upper bound vanishes as $k$ grows large.
Thus, $\lim_{k\to\infty} M_k(A,t) = M_\infty(A,t)$ for each $A$ and $t$. Thus, to show our claim, it is sufficient to show that $1+M_k(A,t)\ge M_k(A+e_i,t)$ for all finite $k$ and all $i\in\agentset$. We now show this via induction on $k$.

This induction hypothesis is trivially true for $k=0$ and so we turn our attention to showing it is true for $k>0$ assuming that it is true for $k-1$. Consider two systems at time $t$, one with state $A$, and the other with state $A+e_i$. We will refer to the two systems as $S_0$ and $S_1$ respectively.

We fix a sequence of agent and job arrivals within $[t,T]$ (i.e., we condition on the sequence of types and the arrival times of agents and jobs). Also consider a collection of independent exponential random variables representing the times until each of several possible reneging events occur in the system. Then the next event occurs randomly within the following categories:
\begin{enumerate}
\item reneging by an agent counted in $A$, from each type with a non-zero component in $A$;
\item an agent arrival, from each of the $\numLoc$ types (the arrival time of next agent is not random but fixed);
\item a job arrival, from each of the $\numLoc$ types (the arrival time of next job is not random but fixed);
\item reneging by the agent not counted in $A$, which is of type $i$ in system $S_1$.
\end{enumerate}
We let random variable $t'$ represent the sum of $t$ and the minimum of these times. The random variable achieving this minimum determines the event that occurs next. If $T$ occurs before any of these times, then the time horizon expires before the next event occurs.

For either system $S_n,\:n=0,1$, let random variable $m(n)$ represent whether a match results from the next event, let $A'(n)$ represent the system state after the event occurs. Let $\Delta$ indicate whether the next event occurs in both systems ($\Delta=1$) or just system $S_1$ ($\Delta=0$, which occurs on event type 4). Thus,
\begin{equation*}
\begin{split}
M_k(A,t) = \E[m(0) + M_{k-\Delta}(A'(0),t')]
\ge \E[m(0) + M_{k-1}(A'(0),t')],\\
M_k(A+e_i,t) = \E[m(1) + M_{k-1}(A'(1),t')].
\end{split}
\end{equation*}

The expectation is taken over $m(n), A'(n)$ and $t'$. To show our result, we will condition on the next event on a case-by-case basis to show that
\begin{equation}
\label{eq:lm4}
1 + m(0) + M_{k-1}(A'(0),t')
\geq_{\textrm{st}} m(1) + M_{k-1}(A'(1),t'),
\end{equation}
where the comparison represents first-order stochastic dominance. If $t'>T$, then $M_{k-1}(A'(n),t')=0, m(n) = 0$ for all $n$, automatically verifying this expression. Thus, it is sufficient to focus on $t'\le T$.

\medskip

\emph{\underline{Case 1}}: The next event is reneging by an agent counted in $A$ or an agent arrival. In this case, let $A''$ represent $A$ modified by this event, so that $A'(0) = A''$ and $A'(1) = A'' + e_i$.
Also, $m(0)=m(1)=0$. Then, by the induction hypothesis,
$1+M_{k-1}(A'', t')
\ge M_{k-1}(A'' + e_i, t')$,
showing \eqref{eq:lm4}.

\medskip

\emph{\underline{Case 2}}: The next event is a job of type $j\notin\agentset(i)$. In this case, $m(0)$ is a Bernoulli random variable with probability $\beta(\sum_{i'\in \agentset(j)} A_{i'}, \sum_{i'\in\agentset}A_{i'})$ being one, and $m(1)$ is a Bernoulli random variable with probability $\beta(\sum_{i'\in \agentset(j)} A_{i'}, \sum_{i'\in\agentset}A_{i'}+1)$ being one. We have $\beta(\sum_{i'\in \agentset(j)} A_{i'}, \sum_{i'\in\agentset}A_{i'})\ge \beta(\sum_{i'\in \agentset(j)} A_{i'}, \sum_{i'\in\agentset}A_{i'}+1)$ by \Cref{lemma:beta}. Thus, there exists a coupling of the randomness of $m(0)$ and $m(1)$ such that: (1) $m(0)=m(1)=1$, both systems match to the same agent; (2) $m(0)=1, m(1)=0$, $S_0$ matches the job to an agent in $A$ while $S_1$ fails to match; (3) $m(0)=m(1)=0$, both systems fail to match the job. %
We now discuss case by case:

\begin{enumerate}
    \item when $m(0)=m(1)=1$, and suppose $i'\in\agentset(j)$ is the agent type both systems match this type $j$ job to. We have $A'' = A - e_{i'}$ representing $A$ modified by this event so that $A'(0) = A''$ and $A'(1) = A'' + e_i$. Then, by the induction hypothesis,
    $2 + M_{k-1}(A'', t')
    \ge 1  + M_{k-1}(A'' + e_i, t')$,
    showing \eqref{eq:lm4};

    \item when $m(0)=1, m(1)=0$, $\mathcal{S}_0$ matches this job to a type $i'\in\agentset(j)$ agent counted in $A$. We have $A'' = A - e_{i'}$ representing $A$ modified by this event so that $A'(0) = A''$ and $A'(1) = A + e_{i'} + e_i$. Then, by the induction hypothesis,
    $2 + M_{k-1}(A'', t')
    \ge M_{k-1}(A'' + e_{s'} + e_i, t'),$
    showing \eqref{eq:lm4};

    \item when $m(0) = m(1) =0$, we have $A'(0)=A$ and $A'(1) = A + e_i$. Then, by the induction hypothesis,
    $1 + M_{k-1}(A, t')
    \ge  M_{k-1}(A + e_{i}, t')$,
    showing \eqref{eq:lm4}.
\end{enumerate}

\medskip

\emph{\underline{Case 3}}: The next event is a job of type $j\in\agentset(i)$. In this case, $m(0)$ is a Bernoulli random variable with probability $\beta(\sum_{i'\in \agentset(j)} A_{i'}, \sum_{i'\in\agentset}A_{i'})$ being one, $m(1)$ is a Bernoulli random variable with probability $\beta(\sum_{i'\in \agentset(j)} A_{i'} + 1, \sum_{i'\in\agentset}A_{i'}+1)$ being one. We have $\beta(\sum_{i'\in \agentset(j)} A_{i'} + 1, \sum_{i'\in\agentset}A_{i'} + 1)\ge \beta(\sum_{i'\in \agentset(j)} A_{i'}, \sum_{i'\in\agentset}A_{i'})$ by \Cref{lemma:beta}. Thus, there exists a coupling of the randomness of $m(0), m(1)$ such that: (1) $m(0)=m(1)=1$, both systems match to the same agent, or $S_1$ matches the job to the additional type $0$ agent $e_i$ not counted in $A$, while $S_0$ matches an agent in $A$; (2) $m(0)=0, m(1)=1$, $S_1$ matches the job to the additional type $i$ agent not counted in $A$ while $S_0$ fails to match the job; (3) both systems fail to match the job. We now discuss all these outcomes case by case:

\begin{enumerate}
    \item when $m(0)=m(1)=1$, and both systems match to the same agent. Suppose $i'\in\agentset(j)$ is the agent type both systems match this type $j$ job to. We have $A'' = A - e_{i'}$ represent $A$ modified by this event so that $A'(0) = A'', A'(1) = A'' + e_i$. Then, by the induction hypothesis,
    $2 + M_{k-1}(A'', t')
    \ge 1  + M_{k-1}(A'' + e_i, t')$,
    showing \eqref{eq:lm4};

    \item when $m(0)=m(1)=1$, and $S_1$ matches the job to the additional type $i$ agent not counted in $A$, $S_0$ matches an agent of type $i'$ in $A$. We have $A'(0) = A-e_i, A'(1) = A$. Then, by the induction hypothesis,
    $2 + M_{k-1}(A-e_i, t')
    \ge 1 + M_{k-1}(A, t')$,
    showing \eqref{eq:lm4};

    \item when $m(0)=0, m(1)=1$, $S_1$ matches the job to the additional type $i$ agent not counted in $A$, while $S_0$ fails to match. We have $A'(0) = A'(1) = A$. Then, \eqref{eq:lm4} can be directly shown;
    
    \item when $m(0) = m(1) = 0$, we have $A'(0)=A$ and $A'(1) = A + e_i$. Then, by the induction hypothesis,
    $1 + M_{k-1}(A, t')
    \ge  M_{k-1}(A + e_{i}, t')$,
    showing \eqref{eq:lm4}.
\end{enumerate}

\medskip

\emph{\underline{Case 4}}: The next event is reneging by the agent not counted in $A$. 
In this case, $A'(0) = A'(1) = A$ and $m(0) = m(1) = 0$.
Thus, \eqref{eq:lm4} is verified directly, with $1+M_{k-1}(A,t')\ge M_{k-1}(A,t')$.

\smallskip

This concludes the proof of \Cref{lm:v_flexibility_2}.~$\square$

\end{proof}

\medskip

\begin{proof}{\textbf{Proof of \Cref{thm:robust_improvement_FLQR}}.}

Fix agent type $i\neq0$. Under the FRfb policy, type $i$ agents can only be matched in queue $i$ or queue $0$ with type $j=i$ jobs. In any other queue their virtual waiting time would be infinity. Since the FRfb policy only attempts to match type $i$ jobs with agents in queue $1$ when queue $i$ becomes empty, a type $i$ agent would wait less in queue $i$ than in queue $0$. This holds regardless of the strategy that flexible agents play. 

\smallskip

We now define two sets of state representations for the FRfb policy and the random policy respectively.
\begin{itemize}
    \item The random policy: $\vecAgent=\left(\agent_i\right)_{i\in\agentset}$, number of agents of each type as all agents join the same queue.  
    
    \item The FRfb policy: $\vecQueue=\left(\{\queue_{0,i}\}_{i\in\agentset}, \{\queue_{i}\}_{i\in\agentset\setminus\{0\}}\right)$ where $\queue_{0,i}$ is the number of type $1$ agents in queue $i$ and $\queue_i$ is the number of type $i\neq0$ agents in queue $i$. %
\end{itemize}

We first prove the result for a fixed time horizon $T>0$.  We want to prove the following inequality holds for all $T\ge0, \joinprob\in\joinprobset(\mathcal{Q}^\RCR)$:
\begin{equation*}\label{eq:key-property}
\E[M^{\FLQR}(T;\joinprob)\mid \vecQueue\supZero = \vecQueue] \geq \E[M^{\SMQ}(T)\mid \vecAgent\supZero = \vecAgent],\:\: \forall\vecQueue,\vecAgent: \agent_0=\sum_{i\in\agentset}\queue_{0,i}, \agent_i=\queue_i, \forall i\in\agentset\setminus\{0\}.
\end{equation*}

Let $M^{\FLQR}_k(Q, t)$ and $M^{\SMQ}_k(A, t)$ be the expected number of matches starting with states $Q$ and $A$ from time $t$ up until the time horizon $T$ expires or until we reach a total of $k$ events, whichever occurs first, under the RCR policy and the random policy respectively. Events include job arrivals, agent arrivals, and reneging. It is sufficient to show that under any strategy profile $\joinprob\in\joinprobset(\mathcal{Q}^\RCR)$, $M^{\FLQR}_k(Q, t)\ge M^{\SMQ}_k(A, t)$ for any finite $k$, any $t\in[0, T]$, any $\vecQueue,\vecAgent$ such that $\agent_0=\sum_{i\in\agentset}\queue_{0,i},\:\agent_i=\queue_i,\:\forall i\in\agentset\setminus\{0\}$. We show this via induction on $k$. It can be seen that $k=0$ holds trivially, so we proceed to show it is true for $k>0$ assuming that it is true for $k-1$.

We fix a sequence of agent and job arrivals within $[t,T]$ (i.e., we condition on the sequence of types and the arrival times of agents and jobs). Also consider a collection of independent exponential random variables representing the times until each of several possible reneging events occur in the system. Then the next event occurs randomly within the following categories:

\begin{enumerate}
\item reneging by an agent counted in $A$, from each type with a non-zero component in $A$;
\item an agent arrival, from a type in $\agentset$ (the arrival time of next agent is not random but fixed);
\item a job arrival, from a type in $\loc$ (the arrival time of next job is not random but fixed).
\end{enumerate}

We let random variable $t'$ represent the sum of $t$ and the minimum of these times. The random variable achieving this minimum determines the event that occurs next. If $T$ occurs before any of these times, then the time horizon expires before the next event occurs.

Under the FRfb policy and the random policy, let random variables $m(\policy^\RCR)$ and $m(\policy^\RND)$ represent whether a match results from the next event, let $Q'(\policy^\RCR)$ and $A'(\policy^\RND)$ represent the system state after the event occurs. Thus,
\begin{equation*}
\begin{split}
M_k^{\policy^\RCR}(Q,t) = \E[m(\policy^\RCR) + M^{\policy^\RCR}_{k-1}(Q'(\policy^\RCR),t')], \\
M_k^{\policy^\RND}(A,t) = \E[m(\policy^\RND) + M^{\policy^\RND}_{k-1}(A'(\policy^\RND),t')].
\end{split}
\end{equation*}

The expectation is taken over $m(\cdot), A'(\cdot), Q'(\cdot)$ and $t'$. To show our result, we will consider the next event on a case-by-case basis to show that with probability one,
\begin{equation}
\label{eq:proof-revisit-theorem-2}
m(\policy^\RCR) + M^{\policy^\RCR}_{k-1}(Q'(\policy^\RCR),t')
\ge m(\policy^\RND) + M^{\policy^\RND}_{k-1}(A'(\policy^\RND),t').
\end{equation}
If $t'>T$, then $m(\policy^\RCR) = m(\policy^\RND) = 0$ and $M^{\policy^\RCR}_{k-1}(Q'(\policy^\RCR),t') = M^{\policy^\RND}_{k-1}(A'(\policy^\RND),t') = 0$, automatically verifying this expression. Thus, it is sufficient to focus on $t'\le T$.

\medskip

\underline{\emph{Case 1}}: The next event is reneging by an agent counted in $A$ or an agent arrival.  In this case, we have $Q'(\policy^\RCR) = A'(\policy^\RND)$ and $m(\policy^\RCR) = m(\policy^\RND) = 0$. Then, the induction hypothesis
shows \eqref{eq:proof-revisit-theorem-2}.

\medskip

\underline{\emph{Case 2}}: The next event is a job of type $j$ and $A_j > 0$. When $j\neq 0$, this implies that $Q_j>0$ and,
\begin{align}
\nonumber
    &m(\policy^\RCR) + M^{\policy^\RCR}_{k-1}(Q'(\policy^\RCR),t') \\[2mm] \nonumber
    =& 1 + \frac{Q_j}{Q_j + Q_{0,j}}M^{\policy^\RCR}_{k-1}(Q - e_j,t') + \frac{Q_{0,j}}{Q_j + Q_{0,j}}M^{\policy^\RCR}_{k-1}(Q - e_{0,j},t')\\[2mm] \label{eq:proof-revisit-theorem-2-0} %
    \ge& 1 + \frac{A_j}{A_0+A_j}M^{\policy^\RCR}_{k-1}(Q - e_j,t') + \frac{A_0}{A_0+A_j}M^{\policy^\RCR}_{k-1}(Q - e_{0,j},t') \\[2mm]  \tag{induction hypothesis} %
    \ge& 1 + \frac{A_j}{A_0+A_j}M^{\policy^\RND}_{k-1}(A - e_j,t') + \frac{A_0}{A_0+A_j}M^{\policy^\RND}_{k-1}(A - e_0,t')  \\ \nonumber
    =& \beta\left(A_0+A_j, \sum_{i\in\agentset}A_i\right)\left(1 + \frac{A_j}{A_0+A_j}M^{\policy^\RND}_{k-1}(A - e_j,t') + \frac{A_0}{A_0+A_j}M^{\policy^\RND}_{k-1}(A - e_0,t') \right)  \\ \nonumber
    &+\left(1 - \beta\left(A_0+A_j, \sum_{i\in\agentset}A_i\right)\right)\left(\frac{A_j}{A_0+A_j}\left(1+M^{\policy^\RND}_{k-1}(A - e_j,t')\right) + \frac{A_0}{A_0+A_j}\left(1+M^{\policy^\RND}_{k-1}(A - e_0,t')\right) \right)   \\ \nonumber
    \ge& \beta\left(A_0+A_j, \sum_{i\in\agentset}A_i\right)\left(1 + \frac{A_j}{A_0+A_j}M^{\policy^\RND}_{k-1}(A - e_j,t') + \frac{A_0}{A_0+A_j}M^{\policy^\RND}_{k-1}(A - e_0,t') \right)  \\ \tag{\Cref{lm:v_flexibility_2}} %
    &+\left(1 - \beta\left(A_0+A_j, \sum_{i\in\agentset}A_i\right)\right)\left(\frac{A_j}{A_0+A_j}M^{\policy^\RND}_{k-1}(A,t') + \frac{A_0}{A_0+A_j}M^{\policy^\RND}_{k-1}(A,t') \right)   \\ \nonumber
    =& \beta\left(A_0+A_j, \sum_{i\in\agentset}A_i\right) + \beta\left(A_0+A_j, \sum_{i\in\agentset}A_i\right)\left(\frac{A_j}{A_0+A_j}M^{\policy^\RND}_{k-1}(A - e_j,t') + \frac{A_0}{A_0+A_j}M^{\policy^\RND}_{k-1}(A - e_0,t') \right)  \\ \nonumber
    &+\left(1 - \beta\left(A_0+A_j, \sum_{i\in\agentset}A_i\right)\right)M^{\policy^\RND}_{k-1}(A,t')   \\ \nonumber
    =&m(\policy^\RND) + M^{\policy^\RND}_{k-1}(A'(\policy^\RND),t'),
\end{align}
showing \eqref{eq:proof-revisit-theorem-2}. Inequality \eqref{eq:proof-revisit-theorem-2-0} holds because: (1) $Q_{0,j}\le A_0,\: Q_j=A_j$ which leads to $Q_j/(Q_j + Q_{0,j})\ge A_j/(A_j + A_0),\: Q_{0,j}/(Q_j + Q_{0,j})\le A_0/(A_j + A_0)$; (2)  $M^{\policy^\RCR}_{k-1}(Q - e_j,t')\ge M^{\policy^\RCR}_{k-1}(Q - e_{0,j},t')$ by \Cref{lm:v_flexibility_2}. %

\smallskip

When $j=0$, we have $\sum_{i\in\agentset}Q_{0,i}>0$. Under the case that $Q_{0,0}>0$,
\begin{align}
\nonumber
    &m(\policy^\RCR) + M^{\policy^\RCR}_{k-1}(Q'(\policy^\RCR),t') \\[2mm] \nonumber
    =& 1 + M^{\policy^\RCR}_{k-1}(Q - e_{0,0},t')  \\[2mm] \tag{induction hypothesis} \ge& 1 + M^{\policy^\RND}_{k-1}(A - e_0,t')  \\[2mm] \nonumber
    =& \beta\left(A_0, \sum_{i\in\agentset}A_i\right)\left(1 + M^{\policy^\RND}_{k-1}(A - e_0,t')\right) + \left(1- \beta\left(A_0, \sum_{i\in\agentset}A_i\right)\right)\left(1 + M^{\policy^\RND}_{k-1}(A - e_0,t')\right) \\ \tag{\Cref{lm:v_flexibility_2}} %
    \ge& \beta\left(A_0, \sum_{i\in\agentset}A_i\right)\left(1+M^{\policy^\RND}_{k-1}(A - e_0,t')\right) + \left(1-\beta\left(A_0, \sum_{i\in\agentset}A_i\right)\right) M^{\policy^\RND}_{k-1}(A,t') \\[2mm] \nonumber
    =&m(\policy^\RND) + M^{\policy^\RND}_{k-1}(A'(\policy^\RND),t'),
\end{align}
showing \eqref{eq:proof-revisit-theorem-2}. %

\smallskip

Under the case that $Q_{0,0}=0$, let random variable $X\in\{0,\cdots,\numAgent\}$ denote the queue number from which the job is matched to a type $0$ agent conditioning on the job being successfully matched,
\begin{align}
\nonumber
    &m(\policy^\RCR) + M^{\policy^\RCR}_{k-1}(Q'(\policy^\RCR),t') \\[2mm] \nonumber
        =& \beta\left(\sum_{i\in\agentset}Q_{0,i},\:\sum_{i\in\agentset}Q_{0,i} + \sum_{i\in\agentset\setminus\{0\}}Q_i\right) \E_X\left[1 + M^{\policy^\RCR}_{k-1}(Q - e_{0,X},t')\right]  \\  \nonumber %
    &+\left(1 - \beta\left(\sum_{i\in\agentset}Q_{0,i},\: \sum_{i\in\agentset}Q_{0,i} + \sum_{i\in\agentset\setminus\{0\}}Q_i\right)\right)M^{\policy^\RCR}_{k-1}(Q,t') \\ \nonumber
    =& \beta\left(A_0, \sum_{i\in\agentset}A_i\right) \E_X\left[ 1 + M^{\policy^\RCR}_{k-1}(Q - e_{0,X},t')\right] +\left(1 - \beta\left(A_0, \sum_{i\in\agentset}A_i\right)\right)M^{\policy^\RCR}_{k-1}(Q,t') \\ \tag{induction hypothesis} %
    \ge& \beta\left(A_0, \sum_{i\in\agentset}A_i\right) \left( 1 + M^{\policy^\RND}_{k-1}(A - e_0,t')\right) + \left(1 - \beta\left(A_0, \sum_{i\in\agentset}A_i\right)\right)M^{\policy^\RND}_{k-1}(A,t') \\[2mm] \nonumber
    =&m(\policy^\RND) + M^{\policy^\RND}_{k-1}(A'(\policy^\RND),t'),
\end{align}
showing \eqref{eq:proof-revisit-theorem-2}. %

\medskip

\emph{\underline{Case 3}}: The next event is a job of type $j\in\loc$ and $A_j = 0$ as well as $A_0=0$. This implies that $Q_j=0$ and $\sum_{i\in\agentset}Q_{0,i}=0$. 
In this case, \eqref{eq:proof-revisit-theorem-2} holds simply by the induction hypothesis. %

\medskip

\emph{\underline{Case 4}}: The next event is a job of type $j\neq0$ and $A_j = 0$ but $A_0>0$. This implies that $Q_j=0$ and $\sum_{i\in\agentset}Q_{0,i}>0$. Under the case that $Q_{0,j}=0$ and $Q_{0,0}=0$, let random variable $X\in\{0,\cdots,\numAgent\}, X\neq j$ denote the queue number from which the job is matched to a type 0 agent conditioning on the job being successfully matched,
\begin{align}
\nonumber
    &m(\policy^\RCR) + M^{\policy^\RCR}_{k-1}(Q'(\policy^\RCR),t') \\[2mm] \nonumber
    =& \beta\left(\sum_{i\in\agentset}Q_{0,i},\: \sum_{i\in\agentset}Q_{0,i} + \sum_{i\in\agentset\setminus\{0\}}Q_i\right) \E_X\left[ 1 + M^{\policy^\RCR}_{k-1}(Q - e_{0,X},t')\right]  \\ \nonumber %
    &+\left(1 - \beta\left(\sum_{i\in\agentset}Q_{0,i},\: \sum_{i\in\agentset}Q_{0,i} + \sum_{i\in\agentset\setminus\{0\}}Q_i\right)\right)M^{\policy^\RCR}_{k-1}(Q,t') \\ \nonumber
    =& \beta\left(A_0, \sum_{i\in\agentset}A_i\right) \E_X\left[ 1 + M^{\policy^\RCR}_{k-1}(Q - e_{0,X},t')\right] + \left(1 - \beta\left(A_0, \sum_{i\in\agentset}A_i\right)\right)M^{\policy^\RCR}_{k-1}(Q,t') \\ \tag{induction hypothesis}  %
    \ge& \beta\left(A_0, \sum_{i\in\agentset}A_i\right) \left( 1 + M^{\policy^\RND}_{k-1}(A - e_0,t')\right) +\left(1 - \beta\left(A_0, \sum_{i\in\agentset}A_i\right)\right)M^{\policy^\RND}_{k-1}(A,t') \\[2mm] \nonumber
    =&m(\policy^\RND) + M^{\policy^\RND}_{k-1}(A'(\policy^\RND),t'),
\end{align}
showing \eqref{eq:proof-revisit-theorem-2}. %

\smallskip

Under the case that $Q_{0,j}>0$, 
\begin{align}
\nonumber
    &m(\policy^\RCR) + M^{\policy^\RCR}_{k-1}(Q'(\policy^\RCR),t') \\[2mm] \nonumber
    =&1 + M^{\policy^\RCR}_{k-1}(Q-e_{0,j},t') \\ \nonumber
    =& \beta\left(A_0,\: \sum_{i\in\agentset}A_i\right) \left( 1 + M^{\policy^\RCR}_{k-1}(Q - e_{0,j},t')\right) +\left(1 - \beta\left(A_0,\: \sum_{i\in\agentset}A_i\right)\right)\left( 1 + M^{\policy^\RCR}_{k-1}(Q - e_{0,j},t')\right) \\ \nonumber
    \ge& \beta\left(A_0,\: \sum_{i\in\agentset}A_i\right) \left( 1 + M^{\policy^\RND}_{k-1}(A - e_{0},t')\right) \\ \tag{induction hypothesis}
    &+\left(1 - \beta\left(A_0,\: \sum_{i\in\agentset}A_i\right)\right)\left( 1 + M^{\policy^\RND}_{k-1}(A - e_{0},t')\right) \\ \tag{\Cref{lm:v_flexibility_2}}
    \ge& \beta\left(A_0,\: \sum_{i\in\agentset}A_i\right) \left( 1 + M^{\policy^\RND}_{k-1}(A - e_0,t')\right) + \left(1 - \beta\left(A_0,\: \sum_{i\in\agentset}A_i\right)\right)M^{\policy^\RND}_{k-1}(A,t') \\[2mm] \nonumber
    =&m(\policy^\RND) + M^{\policy^\RND}_{k-1}(A'(\policy^\RND),t'),
\end{align}
showing \eqref{eq:proof-revisit-theorem-2}. %
Finally, under the case that $Q_{0,j}=0$ but $Q_{0,0}>0$,
\begin{align}
\nonumber
    &m(\policy^\RCR) + M^{\policy^\RCR}_{k-1}(Q'(\policy^\RCR),t') \\[2mm] \nonumber
    =&1 + M^{\policy^\RCR}_{k-1}(Q-e_{0,0},t') \\ \nonumber
    =& \beta\left(A_0,\: \sum_{i\in\agentset}A_i\right) \left( 1 + M^{\policy^\RCR}_{k-1}(Q - e_{0,0},t')\right) +\left(1 - \beta\left(A_0,\: \sum_{i\in\agentset}A_i\right)\right)\left( 1 + M^{\policy^\RCR}_{k-1}(Q - e_{0,0},t')\right) \\ \nonumber
    \ge& \beta\left(A_0,\: \sum_{i\in\agentset}A_i\right) \left( 1 + M^{\policy^\RND}_{k-1}(A - e_{0},t')\right) \\ \tag{induction hypothesis}
    &+\left(1 - \beta\left(A_0,\: \sum_{i\in\agentset}A_i\right)\right)\left( 1 + M^{\policy^\RND}_{k-1}(A - e_{0},t')\right) \\ \tag{\Cref{lm:v_flexibility_2}}
    \ge& \beta\left(A_0,\: \sum_{i\in\agentset}A_i\right) \left( 1 + M^{\policy^\RND}_{k-1}(A - e_0,t')\right) + \left(1 - \beta\left(A_0,\: \sum_{i\in\agentset}A_i\right)\right)M^{\policy^\RND}_{k-1}(A,t') \\[2mm] \nonumber
    =&m(\policy^\RND) + M^{\policy^\RND}_{k-1}(A'(\policy^\RND),t'),
\end{align}
showing \eqref{eq:proof-revisit-theorem-2}. This completes the proof.~$\square$ 

\end{proof}

\medskip

\begin{proof}{\textbf{Proof of \Cref{thm:gap_first_best}.}}
Let $\agentset_T$ be the set of agents that arrive in sample path $\omega_T$ during time $[0,T]$.
For any matching policy $\policy$ define
\begin{flalign*}
    \mathcal{A}_1^{\policy} =\{a\in \agentset_T: a\:\: \text{is matched by } \policy\},\:\:\: \mathcal{A}_2^{\policy} =\{a\in \agentset_T: a\:\: \text{reneges in } \policy\},
\end{flalign*}
where in $\mathcal{A}_2^{\policy}$ we include those agents did not renege in $[0,T]$ but were not matched as well. With a bit of abuse of notation, let $\pi(\omega_T)$ be the total number of matches under sample path $\omega_T$ and policy $\pi$. Let OPT and OPT($\omega_T$) be the offline optimal policy and its number of matches. We have that 
\begin{flalign*}
    \OPT(\omega_T)-\policy(\omega_T)&=|\mathcal{A}_1^{\OPT}|-|\mathcal{A}_1^\policy|\\
    &=|\mathcal{A}_1^{\OPT}\cap \mathcal{A}_1^\policy|+|\mathcal{A}_1^{\OPT}\cap \mathcal{A}_2^\policy|-|\mathcal{A}_1^\policy\cap \mathcal{A}_1^{\OPT}|-|\mathcal{A}_1^\policy\cap \mathcal{A}_2^{\OPT}|\\
    &=|\mathcal{A}_1^{\OPT}\cap \mathcal{A}_2^\policy|-|\mathcal{A}_1^\policy\cap \mathcal{A}_2^{\OPT}|\\
    &\leq |\mathcal{A}_1^{\OPT}\cap \mathcal{A}_2^\policy|.
\end{flalign*}
Consider an agent $a$ in $\mathcal{A}_1^{\OPT}\cap \mathcal{A}_2^{\policy}$, this agent is matched by $\OPT$ at time, say, $t$. At this time $a$ is in both systems, the one run by $\OPT$ and $\policy$. Also, since $\OPT$ is matching $a$ there must be a job arrival at $t$. Policy $\policy$ is not matching this incoming job arrival with $a$ (because $a$ reneges in $\policy$); but since $a$ is in the system run by $\policy$ and $\policy$ is non-idling then it must be that $\policy$ is matching the job arrival to some other agent $a'$. That is, for every $a\in \mathcal{A}_1^{\OPT}\cap \mathcal{A}_2^{\policy}$ there exists another $a'\in \agentset_T$ that is matched by $\policy$. Therefore $|\mathcal{A}_1^{\OPT}\cap \mathcal{A}_2^{\policy}|\le \policy(\omega_T)$. This concludes the proof.~$\square$
\end{proof}

\medskip

\begin{proof}{\textbf{Proof of \Cref{prop:fluid_UB}.}}
Our first-best dynamic stochastic matching problem is a special case of the limited-time dynamic stochastic matching problem studied in \cite{aouad2022dynamic}. In particular, rather than a general compatibility graph, we consider a bipartite structure where types are partitioned into two disjoint sets, $\agentset$ and $\loc$, with matches permitted only between them; moreover, jobs in $\loc$ have zero patience. As a result, our linear program specializes program \textsf{(CB)} in \cite{aouad2022dynamic}, involving only matching variables $(x_{i,j})_{i\in\agentset,\, j\in\loc(i)}$, where $i\in\agentset$ are always the active vertices and $j\in\loc$ the passive vertices in the terminology of \cite{aouad2022dynamic} and \cite{huang2018match}, with \emph{no} variables of the form $x_{j,i}$. The result then follows directly from the proof of Lemma~1 in \cite{aouad2022dynamic}. $\square$
\end{proof}

\end{document}

%% file: figures/compatibility_graph_nested.tikz
\begin{tikzpicture}[baseline=0pt]

\node at (-10.5-2.9, 4.0){\large \textbf{Flexible Agents}};

\def\sf{-2}

\draw[line width=0.6mm] (-0,4) circle (5.5mm);
\draw[line width=0.6mm] (-0,2) circle (5.5mm);
\draw[line width=0.6mm] (-0,-1) circle (5.5mm);

\draw[fill=gray, gray, opacity=0.2] (-0,4) circle (5.5mm);
\draw[fill=gray, gray, opacity=0.2] (-0,2) circle (5.5mm);
\draw[fill=gray, gray, opacity=0.2] (-0,-1) circle (5.5mm);

\draw[line width=0.5mm] (-0-7.5,4) circle (5.5mm);
\draw[line width=0.5mm] (-0-7.5,2) circle (5.5mm);
\draw[line width=0.5mm] (-0-7.5,-1) circle (5.5mm);

\draw[fill=gray, gray, opacity=0.2] (-0-7.5,4) circle (5.5mm);
\draw[fill=gray, gray, opacity=0.2] (-0-7.5,2) circle (5.5mm);
\draw[fill=gray, gray, opacity=0.2] (-0-7.5,-1) circle (5.5mm);

\draw[line width=0.5mm] (-0-6.5,4)--(-1,4);
\draw[line width=0.5mm] (-0-6.5,2)--(-1,2);
\draw[line width=0.5mm] (-0-6.5,-1)--(-1,-1);

\draw[line width=0.5mm] (-0-6.5,4)--(-1,-1);
\draw[line width=0.5mm] (-0-6.5,4)--(-1,2);

\node at (4.5,1.5){\large \textbf{Jobs}};

\def\sf{2}
\draw node at (1.0+\sf-1,4) {\large type $0$};
\draw node at (1.0+\sf-1,2) {\large type $1$};
\draw node at (1.0+\sf-1,-1) {\large type $J$};

\node at (0,0.5){\Huge $\vdots$};

\node at (-7.5,0.5){\Huge $\vdots$};

\draw node at(-10.25+0.8,-1) {\large type $J$};
\draw node at(-10.25+0.8,4) {\large type $0$};

\draw node at(-10.25+0.8,2) {\large type $1$};

\draw [decorate,decoration={brace,amplitude=10pt},xshift=-4pt,yshift=0pt]
(-10.5,-1) -- (-10.5,2) node [black,midway,xshift=-3.0cm] 
{\large \textbf{Specialized Agents}};

\end{tikzpicture}

%% file: figures/MQ.tikz
\begin{tikzpicture}[baseline=0pt]

\node at (-10.5-3.3,4.0){\large \textbf{Flexible Agents}};
\draw[line width=0.6mm] (-7,4.5)--(-4,4.5);
\draw[line width=0.6mm] (-7,3.5)--(-4,3.5);
\draw[line width=0.6mm] (-4,3.5)--(-4,4.5);

\draw[line width=0.6mm,-] (-3.6,4.0)--(-1,4.0);
\draw[line width=0.6mm,->] (-1,4.0)--(-2.3,4.0);
\draw[line width=0.6mm,dash dot,-] (-3.6,4.0)--(-1,2.0);
\draw[line width=0.6mm,dashed,->] (-2.1666,2.87)--(-2.18,2.88);
\draw[line width=0.6mm,dash dot,-] (-3.6,4.0)--(-1,-1.0);
\draw[line width=0.6mm,dashed,->] (-1.999481,0.924) -- (-2,0.925);
\draw[line width=0.6mm,-] (-3.6,-1.0)--(-1,-1.0);
\draw[line width=0.6mm,->] (-1,-1.0)--(-2.3,-1.0);
\draw[line width=0.6mm,-] (-3.6,2.0)--(-1.0,2.0);
\draw[line width=0.6mm,->] (-1.0,2.0)--(-2.3,2.0);
\draw[line width=0.6mm,->] (-5.5,3.7)--(-5.5,3.0) node at (-5.15,3.0) {\Large $\renege$};
\draw[line width=0.6mm,->] (-5.5,1.7)--(-5.5,1.0) node at (-5.15,1.0) {\Large $\renege$};
\draw[line width=0.6mm,->] (-5.5,-1.3)--(-5.5,-2.0) node at (-5.15,-2.0) {\Large $\renege$};

\draw[fill=gray, gray, opacity=0.2] (-7.0, -1.5) rectangle (-4.0, -0.5);
\draw[fill=gray, gray, opacity=0.2] (-7.0, 1.5) rectangle (-4.0, 2.5);
\draw[fill=gray, gray, opacity=0.2] (-7.0, 3.5) rectangle (-4.0, 4.5);

\def\sf{-2}
\draw[line width=0.6mm] (-7,4.5+\sf)--(-4,4.5+\sf);
\draw[line width=0.6mm] (-7,3.5+\sf)--(-4,3.5+\sf);
\draw[line width=0.6mm] (-4,3.5+\sf)--(-4,4.5+\sf);

\draw[line width=0.6mm] (-7,0.5+\sf)--(-4,0.5+\sf);
\draw[line width=0.6mm] (-7,1.5+\sf)--(-4,1.5+\sf);
\draw[line width=0.6mm] (-4,0.5+\sf)--(-4,1.5+\sf);

\draw[line width=0.6mm] (-0,4) circle (5.5mm);
\draw[line width=0.6mm] (-0,2) circle (5.5mm);
\draw[line width=0.6mm] (-0,-1) circle (5.5mm);

\draw[fill=gray, gray, opacity=0.2] (-0,4) circle (5.5mm);
\draw[fill=gray, gray, opacity=0.2] (-0,2) circle (5.5mm);
\draw[fill=gray, gray, opacity=0.2] (-0,-1) circle (5.5mm);

\node at (4.5,1.5){\large \textbf{Jobs}};

\def\sf{2}
\draw[line width=0.6mm,->] (0+\sf,2.0)--(-1+\sf,2.0) node at (1.0+\sf,2) {\Large $\jobarr_1$};
\draw[line width=0.6mm,->] (0+\sf,4.0)--(-1+\sf,4.0) node at (1.0+\sf,4) {\Large $\jobarr_0$};
\draw[line width=0.6mm,->] (0+\sf,-1.0)--(-1+\sf,-1.0) node at (1.0+\sf,-1) {\Large $\jobarr_\numLoc$};

\node at (0,0.5){\Huge $\vdots$};
\node at (-5.5,0.5){\Huge $\vdots$};
\node at (-9.0,0.5){\Huge $\vdots$};

\draw node at(-10.25+1.5,4.4) {\large $\joinprob_{0,0}$};
\draw node at(-10.25+1.5,3.5) {\large $\joinprob_{0,1}$};
\draw node at(-10.25+0.25,3) {\large $\joinprob_{0,\numAgent}$};
\draw node at(-10.25+1.5,2.3) {\large $\joinprob_{1,1}$};
\draw node at(-10.25+1.5,-0.7) {\large $\joinprob_{\numAgent,\numAgent}$};

\draw[line width=0.6mm,opacity=1.0] (-10.0,4)--(-7.5,-1);
\draw[line width=0.6mm,->,opacity=1.0] (-10.0,4)--(-10+2.5/4.2, 4-5/4.2);
\draw[line width=0.6mm,opacity=1.0] (-10.0,4)--(-7.5,2);
\draw[line width=0.6mm,->,opacity=1.0] (-10.0,4)--(-10+2.5/2.5,4-2/2.5);

\draw[line width=0.6mm,opacity=1.0] (-10.0,4.0)--(-7.5,4.0);
\draw[line width=0.6mm,->,opacity=1.0] (-10.0,4.0)--(-8.7,4.0);

\draw[line width=0.6mm,opacity=1.0] (-10.0,2.0)--(-7.5,2);
\draw[line width=0.6mm,->,opacity=1.0] (-10.0,2.0)--(-8.7,2);
\draw node at(-10.25+1-1.5,2) {\Large $\agentarr_1$};
\draw[line width=0.6mm,opacity=1.0] (-10.0,-1)--(-7.5,-1) ;
\draw[line width=0.6mm,->,opacity=1.0] (-10.0,-1)--(-8.7,-1) ;

\draw node at(-10.25+1-1.5,-1) {\Large $\agentarr_\numAgent$};
\draw node at(-10.25+1-1.5,4) {\Large $\agentarr_0$};

\draw node at(-10.25+4.8,4) {\large queue $0$};
\draw node at(-10.25+4.8,-1.0) {\large queue $\numLoc$};
\draw node at(-10.25+4.8,2.0) {\large queue $1$};

\draw [decorate,decoration={brace,amplitude=10pt},xshift=-4pt,yshift=0pt]
(-11,-1) -- (-11,2) node [black,midway,xshift=-3.0cm] 
{\large \textbf{Specialized Agents}};
\end{tikzpicture}

%% file: figures/braess.tikz
\begin{tikzpicture}[baseline=0pt]

\node at (-10.5-3.3,4.0){\large \textbf{Flexible Agents}};
\node at (-10.5-3.7,2.0){\large \textbf{Specialized Agents}};
\draw[line width=0.6mm] (-7,4.5)--(-4,4.5);
\draw[line width=0.6mm] (-7,3.5)--(-4,3.5);
\draw[line width=0.6mm] (-4,3.5)--(-4,4.5);
\draw[line width=0.6mm,-] (-3.6,4.0)--(-1,4.0);
\draw[line width=0.6mm,<-] (-2.3,4.0)--(-1,4.0);
\draw[line width=0.6mm,dash dot,-] (-3.6,4.0)--(-1,2.0);
\draw[line width=0.6mm,dash dot,<-] (-2.95,3.5)--(-2.937,3.49);
\draw[line width=0.6mm,dash dot,-,blue] (-3.6,2.0)--(-1,4.0);
\draw[line width=0.6mm,dash dot,<-,blue] (-2.95,2.5)--(-2.937,2.51);
\draw[line width=0.6mm,-] (-3.6,2.0)--(-1.0,2.0);
\draw[line width=0.6mm,<-] (-2.3,2.0)--(-1.0,2.0);

\draw node at(-10.25+1.8,4.4) {\large $\joinprob_{0,0}$};
\draw node at(-10.25+1.8,3.4) {\large $\joinprob_{0,1}$};
\draw node at(-10.25+1.8,2.4) {\large $\joinprob_{1,1}$};
\draw node at(-10.25+1-1,2) {\Large $\agentarr_1$};
\draw node at(-10.25+1-1,4) {\Large $\agentarr_0$};
\draw[line width=0.6mm,opacity=1.0] (-9.5,4.0)--(-7.5,2.0); 
\draw[line width=0.6mm,->,opacity=1.0] (-9.5,4.0)--(-7.5-1.0,2.0+1.0); 
\draw[line width=0.6mm,opacity=1.0] (-9.5,4.0)--(-7.5,4.0); %
\draw[line width=0.6mm,->,opacity=1.0] (-9.5,4.0)--(-8.5,4.0); 
\draw[line width=0.6mm,opacity=1.0] (-9.5,2)--(-7.5,2); %
\draw[line width=0.6mm,->,opacity=1.0] (-9.5,2)--(-8.5,2);
\draw[line width=0.6mm,->] (-5.5,3.7)--(-5.5,3) node at (-5.15,3.0) {\Large $\renege$};
\draw[line width=0.6mm,->] (-5.5,1.7)--(-5.5,1) node at (-5.15,1.0) {\Large $\renege$};

\def\sf{-2}
\draw[line width=0.6mm] (-7,4.5+\sf)--(-4,4.5+\sf);
\draw[line width=0.6mm] (-7,3.5+\sf)--(-4,3.5+\sf);
\draw[line width=0.6mm] (-4,3.5+\sf)--(-4,4.5+\sf);

\draw[line width=0.6mm] (-0,4) circle (5.5mm);
\draw[line width=0.6mm] (-0,2) circle (5.5mm);

\node at (4.5,3){\large \textbf{Jobs}};

\def\sf{2}
\draw[line width=0.6mm,->] (0+\sf,2.0)--(-1+\sf,2.0) node at (1.0+\sf,2) {\Large $\jobarr_1$};
\draw[line width=0.6mm,->] (0+\sf,4.0)--(-1+\sf,4.0) node at (1.0+\sf,4) {\Large $\jobarr_0$};

\draw node at(-10.25+4.8,4) {\large queue $0$};
\draw node at(-10.25+4.8,2.0) {\large queue $1$};

\draw[fill=gray, gray, opacity=0.2] (-7.0, 1.5) rectangle (-4.0, 2.5);
\draw[fill=gray, gray, opacity=0.2] (-7.0, 3.5) rectangle (-4.0, 4.5);

\draw[fill=gray, gray, opacity=0.2] (-0,4) circle (5.5mm);
\draw[fill=gray, gray, opacity=0.2] (-0,2) circle (5.5mm);
\end{tikzpicture}

%% file: figures/Lemma4.tikz
\begin{tikzpicture}[baseline=0pt]
\draw[line width=0.6mm,->] (0,0) -- (10.0,0);
\draw[line width=0.6mm,->] (0,0) -- (0,7.5);
\node at (10.0,-0.5){\Large $\sigma_{0,1}$};
\node at (-1.1,7.5){\large Utility};

\draw[line width=0.3mm] (0,3.5) parabola (9,2);

\node at (10,2.1){$u^{\FLQ}_{0,1}(\joinprob_{0,1})$};

\draw[line width=0.3mm, dashed] (0,4.2) .. controls (4.0,4.1) and (8,3) .. (9,3.5);

\node at (10.1,3.5){$u^{\FLQR}_{0,1}(\joinprob_{0,1})$};

\node at (-1.1,4.2){$u^{\FLQR}_{0,1}(0)$};

\node at (-1,3.5){$u^{\FLQ}_{0,1}(0)$};

\draw[line width=0.3mm, color=blue] (0,5.5) .. controls (2,5.0) .. (9,6.2);
\draw[line width=0.3mm, color=blue] (0,3.2) .. controls (2,2.7) .. (9,3.9);
\draw[line width=0.3mm, color=blue] (0,1.0) .. controls (2,0.5) .. (9,1.7);

\draw (0.2,3.5) -- (-0.2,3.5);

\draw (0.2,4.2) -- (-0.2,4.2);

\node[color=blue] at (4.0,6.1){$u^{\FLQR}_{0,0}(\joinprob_{0,1})=u^{\FLQ}_{0,0}(\joinprob_{0,1})$};
\node[color=blue] at (8.5,6.5){\large Case 3};
\node[color=blue] at (8.5,4.2){\large Case 2};
\node[color=blue] at (8.5,1.2){\large Case 1};
\end{tikzpicture}